\pdfoutput=1

\documentclass[journal]{IEEEtran}
\usepackage{cite}
\usepackage{graphicx}  
\usepackage{amsfonts} 
\usepackage{amsmath}
\usepackage{amsthm}
\usepackage{amssymb}
\usepackage{color}
\usepackage{float}
\usepackage{algorithm}  
\usepackage{algpseudocode}  
\usepackage{amsmath}  
\usepackage{url}
\usepackage{graphicx}
\usepackage{subfigure}
\usepackage[font=small]{caption}
\usepackage{booktabs}
\usepackage{multirow}
\usepackage{mathtools}

\newtheorem{mypro}{Proposition}
\newtheorem{lemma}{Lemma}
\newtheorem{corollary}{Corollary}

\begin{document}

\title{Reconfigurable Intelligent Surface Aided Constant-Envelope Wireless Power Transfer}

\author{
	\IEEEauthorblockN{Huiyuan Yang, Xiaojun Yuan,~\IEEEmembership{Senior Member,~IEEE}, Jun Fang,~\IEEEmembership{Senior Member,~IEEE}, and Ying-Chang Liang,~\IEEEmembership{Fellow,~IEEE}}
\thanks{This work was presented in part in IEEE Globecom 2020 \cite{yang2020reconfigurable}.}
\thanks{The authors are with the National Key Laboratory of Science and Technology on Communication, University of Electronic Science and Technology of China, Chengdu 611731, China (e-mail: hyyang@std.uestc.edu.cn; xjyuan@uestc.edu.cn; JunFang@uestc.edu.cn; liangyc@ieee.org). }
}

\maketitle
\begin{abstract}

By reconfiguring the propagation environment of electromagnetic waves artificially, reconfigurable intelligent surfaces (RISs) have been regarded as a promising and revolutionary hardware technology to improve the energy and  spectrum efficiency of wireless networks. In this paper, we study a RIS aided multiuser multiple-input multiple-output (MIMO) wireless power transfer (WPT) system, where the transmitter is equipped with a constant-envelope analog beamformer. First, we maximize the total received power of the users by jointly optimizing the beamformer at transmitter and the phase-shifts at the RIS, and propose two alternating optimization based suboptimal solutions by leveraging the semidefinite relaxation (SDR) and the successive convex approximation (SCA) techniques respectively. Then, considering the user fairness, we formulate another problem to maximize the total received power subject to the users' individual minimum received power constraints. A low complexity iterative algorithm based on both alternating direction method of multipliers (ADMM) and SCA techniques is proposed to solve this problem. In the case of multiple users, we further analyze the asymptotic performance as the number of RIS elements approaches infinity, and bound the performance loss caused by RIS phase quantization. Numerical results show the correctness of the analysis results and the effectiveness of the proposed algorithms.

\end{abstract}


%
\IEEEpeerreviewmaketitle

\section{Introduction}

The proliferation of wireless devices has brought a lot of convenience to modern lives, but their limited battery life requires frequent battery replacement/recharging, which is costly and even infeasible in many applications. By powering wireless devices with radio frequency (RF) energy over the air using a dedicated power transmitter, wireless power transfer (WPT) technology provides an attractive solution \cite{bi2015wireless,zeng2017communications}.


The low efficiency of WPT for users over long distances has been regarded as the performance bottleneck in practical systems \cite{zeng2017communications}. To alleviate this problem, array-based energy beamforming techniques are widely used to obtain a beamforming gain by concentrating the energy of the radiated electromagnetic waves in a particular direction \cite{yang2015throughput,wang2017beamforming,xu2014multiuser}. However, conventional digital beamforming requires that each antenna has its own radio frequency chain, which is costly especially when the antenna number becomes very large \cite{zhang2014massive}. One way to reduce this cost is to apply analog beamforming, in which multiple transmit antennas share a common radio frequency chain with a radio frequency frontend to control the amplitude and phase of the signal at each antenna \cite{hur2013millimeter,zhang2017multi}. In this paper, we consider constant-envelope analog beamforming, which reduces the frontend to a set of phase shifters (one for each transmit antenna), thereby further reducing the hardware cost \cite{zhang2017multi,zhang2014massive,mohammed2013per}. More importantly, constant-envelope signaling allows the transmitter to use non-linear RF power amplifiers, which can significantly improve the energy efficiency\footnote{For wireless power transfer, ``energy efficiency'' refers to the ratio of the total received power of the users to the power consumed by the transmitter.} of the system since power-efficient radio frequency components are generally non-linear \cite{mohammed2013per,mancuso2011reducing}.

Another promising way to increase the efficiency of WPT is to place reconfigurable intelligent surfaces (RISs) in the wireless propagation environment so as to recofigure the propagation environment of electromagnetic waves artificially. RISs is a kind of programmable and reconfigurable passive meta-surfaces which are able to change the transmission direction of the signals \cite{yu2019miso, liang2019large, Chen2019Intelligent, Guo2020Weighted}. Typically, a RIS is a planar array composed of low-cost passive elements, e.g., printed dipoles, where each of the elements can independently reflect incident electromagnetic waves with a adjustable phase-shift to change the reflected signal propagation collaboratively \cite{wu2019intelligent,yuan2020reconfigurable}. In this way, IRSs are able to create favorable wireless propagation environments through intelligent placement and passive beamforming \cite{yuan2020reconfigurable, liang2019large}, so as to increase the efficiency of wireless power transfer.

Recently, the new research paradigm of RIS aided wireless power transfer has been intensively studied \cite{Mishra2019Channel, Wu2020Weighted, Pan2020Intelligent, Wu2020Joint}. Specifically, in \cite{Mishra2019Channel}, the authors proposed a channel estimation scheme and a low-complexity beamforming algorithm in a RIS aided multiple-input single-output (MISO) WPT system. \cite{Wu2020Weighted} and \cite{Pan2020Intelligent} respectively studied the weighted sum-power maximization and weighted sum-rate maximization problems in the RIS aided simultaneous wireless information and power transfer (SWIPT) system, and WPT was also studied as a special case of SWIPT. In a RIS-aided SWIPT system, the authors of \cite{Wu2020Joint} studied the problem of minimizing the total transmitting power with QoS constraints and proposed an effective iterative algorithm.

\begin{figure}[h]
	\centerline{\includegraphics[width=1\columnwidth]{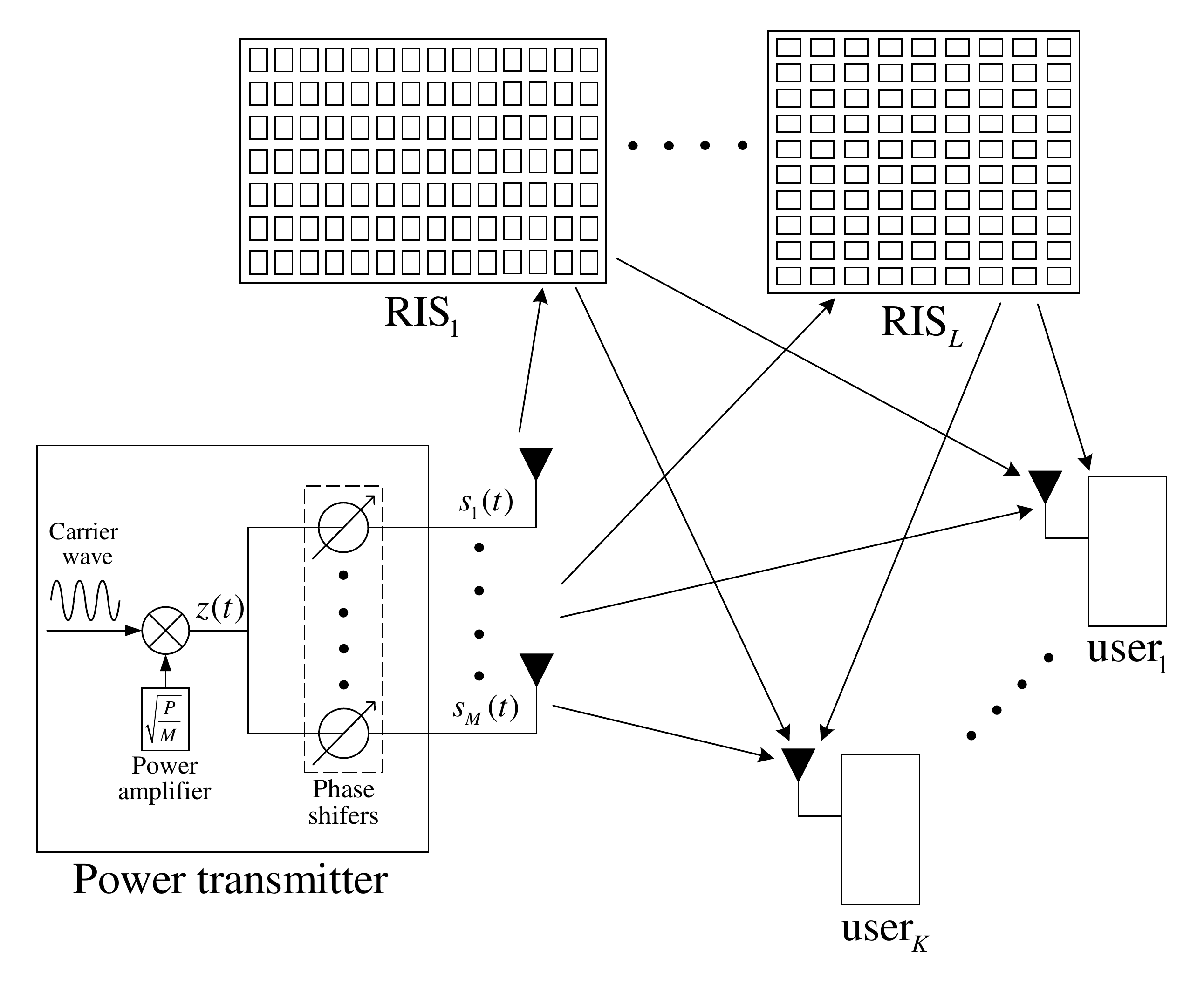}}
	\caption{ RIS aided wireless power transfer in a multiuser MIMO system with the constant-envelope analog beamforming at the multi-antenna transmitter.
	}
\end{figure}

This paper is focused on a multiuser multiple-input multiple-output (MIMO) WPT system aided by several RISs. As shown in Fig. 1, the system consists of $L$ RISs, $K$ single-antenna receivers and a transmitter equipped with $M$ antennas. Following the idea of constant-envelope analog beamforming, all the transmitter antennas share a common radio frequency chain. Therefore, the envelopes of the signals radiated by the transmitting antennas are the same. To the best of our knowledge, this is the first work to study the joint optimization of RIS and transmitter beamforming in the RIS aided constant-envelope WPT system. A typical application of the considered system is to charge the sensors of wireless sensor networks (WSNs). Since these sensors are often placed in hard-to-reach locations, replacing their batteries manually can be inconvenient and even dangerous. WPT provides an attractive solution by remotely recharging the batteries. In addition, since the ubiquitous and uncontrolled barriers may block the propagation path of electromagnetic waves from transmitter to receiver, IRS can greatly increase the energy efficiency of the system by providing new paths. As illustrated in Fig. 2(a), a mobile charger moves in the sensor layout area to remotely charge the batteries of the sensors. Another application is to change mobile devices in office areas. As illustrated in Fig. 2(b), a fixed wireless charger charges multiple mobile devices in the office simultaneously.

\begin{figure}[htb]
	\centerline{\subfigure[A mobile charger with fixed devices]{
			\begin{minipage}[t]{0.5\columnwidth}
				\centerline{\includegraphics[width=\columnwidth]{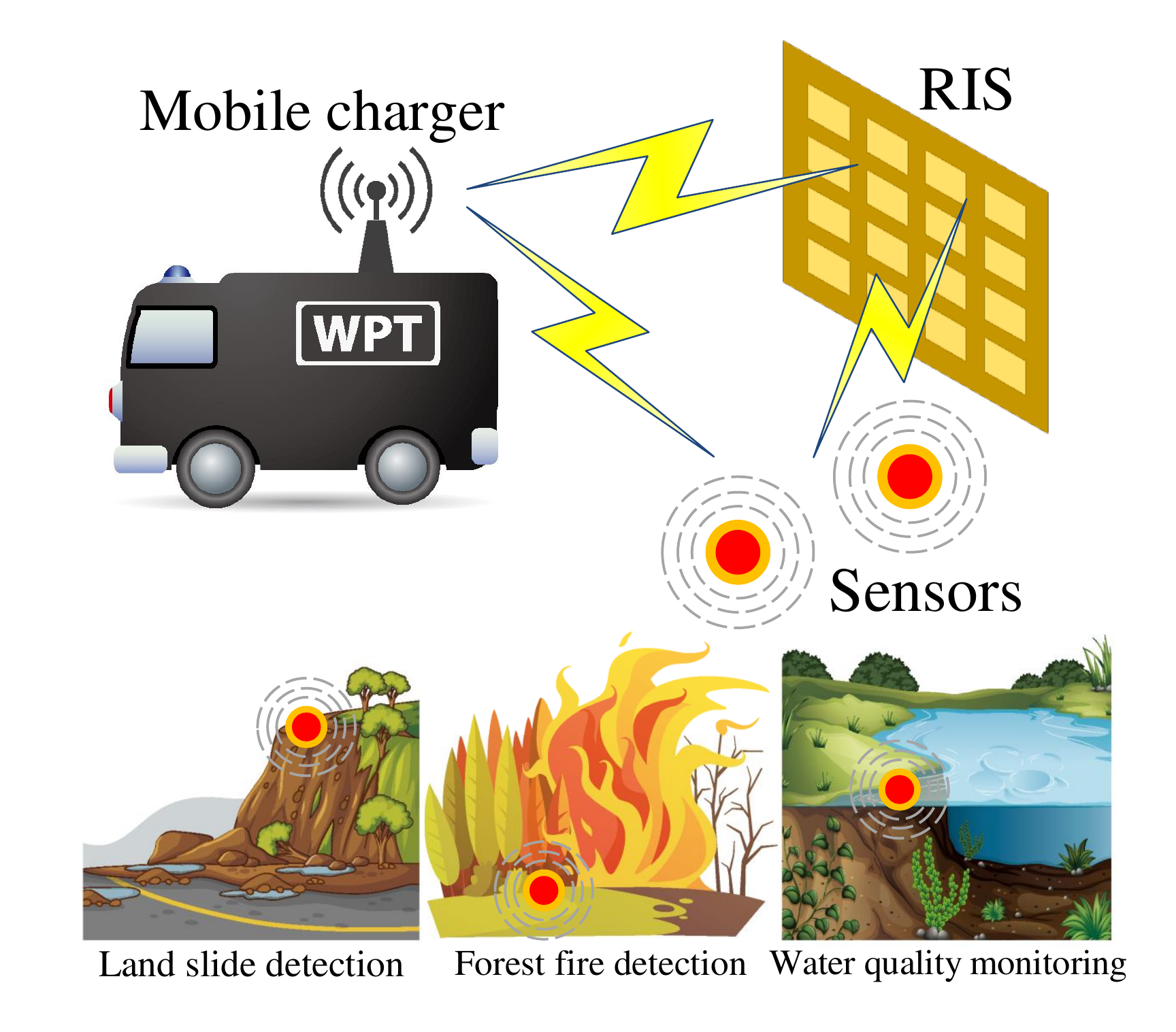}}
			\end{minipage}%
		}%
		\subfigure[A fixed charger with mobile devices]{
			\begin{minipage}[t]{0.5\columnwidth}
				\centerline{\includegraphics[width=\columnwidth]{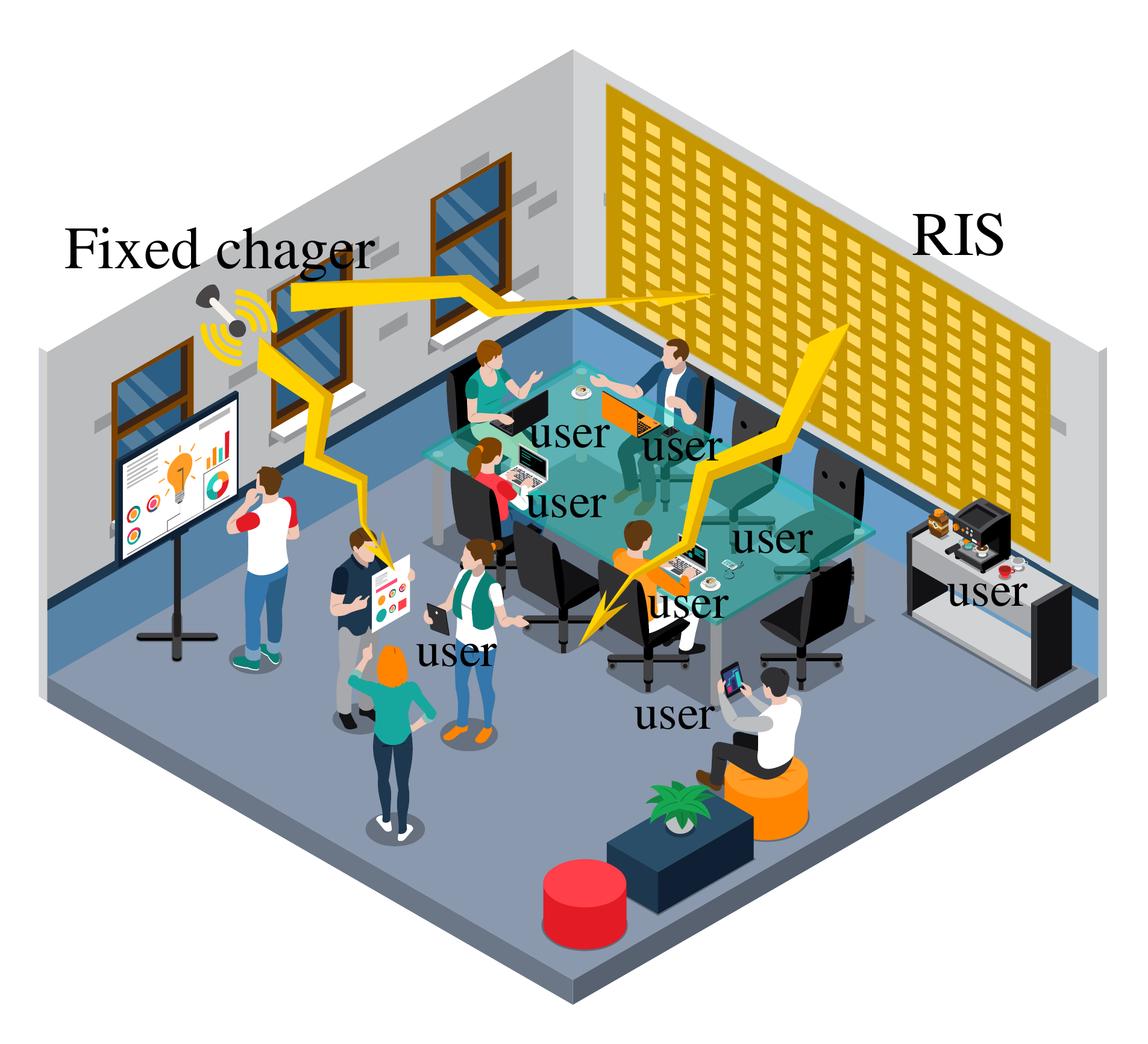}}
			\end{minipage}%
	}}
	\caption{Two applications of RIS aided wireless power transfer.}
\end{figure}

In this paper, we propose two optimization problems that are applicable to different scenarios. The first one is to maximize the total received power. This problem is reasonable in some scenarios, such as in Fig. 2(a). In this scenario, each user has roughly an equal chance to access the charger due to the mobility of the charger, and since each sensor can work for days and even weeks after it is fully charged, the charging service is not delay-sensitive. Therefore, the total received power is an appropriate performance evaluation indicator in such a system. Combined with the physical properties of the transmitter and RIS, our first optimization problem is thus to maximize the total received power of the users by jointly optimizing the beamformer at transmitter and the phase-shifts at the RIS, subject to the constant-envelope constraints which include two parts: 1) the signals radiated by all the transmitting antennas are required to have the same constant envelope; 2) each element of the RIS generates a fixed gain on the reflected signal. One feature of this optimization problem is that the optimal transmitter and RIS beamformers focus electromagnetic energy in the direction of nearby users, at the expense of the quality-of-service (QoS) degradation for far away users. However, the charging service is delay-sensitive in many applications, such as the scenario in Fig. 2(b), in which the mobile devices cannot stay in the charging area for long periods of time, so the charging service of each device must be completed in a relatively short time. Further, different devices may have different requirements on charging speed due to the difference of remaining power or other reasons. This inspires us to consider a new optimization problem that takes into account both energy efficiency and user fairness, i.e., to maximize the total received power of the users subject to the constant-envelope constraints and a minimum QoS constraint for each user.

The main contributions of this paper are summarized as follows: 

\begin{itemize}
	\item We propose a novel energy-efficient RIS aided constant-envelope wireless power transfer scheme.
	
	\item For delay-insensitive applications, we formulate the total received power maximization problem, which is non-convex and difficult to solve optimally. We first develop a semidefinite relaxation (SDR) based algorithm to provide an approximate solution to this problem. To reduce computational complexity, we further design a successive convex approximation (SCA) based low-complexity algorithm.
	
	\item For delay-sensitive applications, we formulate the problem of total received power maximization with QoS constraints. This problem is again non-convex and difficult to solve optimally. To tackle this problem, we combine the ideas of alternating optimization (AO), SCA and alternating direction method of multipliers (ADMM) to develop a effective algorithm with low-complexity.
	
	\item The performance limit of the RIS aided multiuser WPT system as the number of RIS elements $N \to \infty$ is analyzed. We show that, if the RIS is large enough, the average total received power of the $K$ users scales in the order of $N^2$. 
	
	\item The performance loss caused by RIS phase quantization is analyzed. We derive the explicit expression of an upper bound on this performance loss, which shows that, even if the phases can only be quantified to a few given values (e.g. four values), the performance loss caused by quantization is acceptable.
	
\end{itemize}

The rest of this paper is organized as follows. Section $\textrm{\uppercase\expandafter{\romannumeral2}}$ introduces the system model and the problems proposed in this paper. In Section $\textrm{\uppercase\expandafter{\romannumeral3}}$, we consider the maximization of the total received power. In Section $\textrm{\uppercase\expandafter{\romannumeral4}}$, we address the maximization of the total received power under user fairness consideration. The performance analysis is presented in Section $\textrm{\uppercase\expandafter{\romannumeral5}}$. Section $\textrm{\uppercase\expandafter{\romannumeral6}}$ presents the numerical results. Finally, conclusions are drawn in Section $\textrm{\uppercase\expandafter{\romannumeral7}}$.

$Notation$: Scalars, vectors and matrices are denoted by italic letters, bold-face lower-case letters and bold-face upper-case letters, respectively. $d^{\star}$ is the conjugate of $d$. $\textrm{arg}(d)$ denotes the argument of complex number $d$ (set $\textrm{arg}(0)=0$). diag$(\boldsymbol{d})$ denotes a diagonal matrix with each diagonal element being the corresponding element in $\boldsymbol{d}$. $[\boldsymbol{d}]_n$ denotes the $n$-th element in vector $\boldsymbol{d}$. $[\boldsymbol{d}]_{(M:N)}$ denotes the vector that contains from the $M$-th element to the $N$-th element of vector $\boldsymbol{d}$ and  $\textrm{arg}(\boldsymbol{d})=[\text{arg}([\boldsymbol{d}]_1),\dots,\text{arg}([\boldsymbol{d}]_N)]^T$. Let $\textrm{exp}\left(\boldsymbol{d}\right)=[e^{[\boldsymbol{d}]_1},\dots,e^{[\boldsymbol{d}]_N}]^T$. $\|\boldsymbol{d}\|$ represents the Frobenius norm of vector $\boldsymbol{d}$. diag$(\boldsymbol{D})$ denotes a column vector with each element being the corresponding diagonal element in $\boldsymbol{D}$. $[\boldsymbol{D}]_{n,n}$ denotes the element of matrix $\boldsymbol{D}$ in row $n$ and column $n$. $\boldsymbol{D} \succeq \boldsymbol{0}$ means that matrix $\boldsymbol{D}$ is positive semidefinite. $\boldsymbol{D}^{\star}$, $\boldsymbol{D}^T$ and $\boldsymbol{D}^H$ denote the conjugate, the transpose and the conjugate transpose of matrix $\boldsymbol{D}$, respectively. $\textrm{rank}(\boldsymbol{D})$ represents the rank of matrix $\boldsymbol{D}$. $\textrm{tr}(\boldsymbol{D})$ denotes the trace of matrix $\boldsymbol{D}$. $\mathbb{C}^{a \times b}$ denotes the space of $a \times b$ complex-valued matrices. $j$ denotes the imaginary unit. $\boldsymbol{I}_N$ denotes the $N\times N$ identity matrix. $\mathcal{C}\mathcal{N}(\boldsymbol{0},\boldsymbol{I}_N)$ represents the circularly symmetric complex Gaussian (CSCG) distribution with zero mean and covariance matrix $\boldsymbol{I}_N$.

\section{System Model And Problem Formulation}

\subsection{System Model}

Consider a multiuser MIMO wireless power transfer system assisted by a number of reconfigurable intelligent surfaces, where a base station (BS) with $M$ antennas transfers power to $K$ single-antenna\footnote{Most of the discussions in this paper can be literally extended to the case of multi-antenna users by treating each user antenna as an individual ``user''. The only exception is that the minimum received power constraint (13) in (P2) needs to be modified to cover all the antennas of each user. The corresponding algorithm for solving (P2) also calls for slight modifications to accommodate the change of (13).} users simultaneously. The system consists of $L$ RISs, where the $l$-th RIS, denoted by $\textrm{RIS}_l$, has $N_l$ elements, $\forall l$. Since the path loss is significant, we assume that the power of the signals reflected by the RIS twice or more is negligible and thus ignored \cite{wu2019towards}. All the RISs are connected by a controller for reflecting phase adjustment. We assume that perfect channel state information (CSI) is available to the BS, the BS computes the optimal phase coefficients of the RISs, and the BS sends these coefficients to the RIS controller via a separate wired or wireless link. In practice, the CSI can be acquired, e.g., by employing the channel estimation methods in \cite{He2020Cascaded, Wang2020Channel, chen2019channel, Liu2020Matrix}.


The baseband equivalent channels of the BS-user link, the $\textrm{RIS}_l$-user link and the BS-$\textrm{RIS}_l$ link are denoted by $\boldsymbol{H}_{d} \in \mathbb{C}^{K \times M}$, $\boldsymbol{H}_{r,l} \in \mathbb{C}^{K \times N_l}$ and $\boldsymbol{S}_l \in \mathbb{C}^{N_l \times M}$, respectively. Let $N=\sum_{l=1}^{L}N_l$ be the total number of reflecting elements. The baseband equivalent channels of the RIS-user link and the BS-RIS link are denoted by $\boldsymbol{H}_{r} \in \mathbb{C}^{K \times N}$ and $\boldsymbol{S} \in \mathbb{C}^{N \times M}$, respectively, where $\boldsymbol{H}_r=[\boldsymbol{H}_{r,1},\dots,\boldsymbol{H}_{r,L}]$ and $\boldsymbol{S}=[\boldsymbol{S}_1^T,\dots,\boldsymbol{S}_L^T]^T$. 

The gain vector of the $l$-th RIS is denoted by $\boldsymbol{\phi}_l=[\beta_{l,1} e^{j\theta_{l,1}},\ldots,\beta_{l,N_l} e^{j\theta_{l,N_l}}]^T$, where $\beta_{l,n} \in [0,1]$ and $\theta_{l,n} \in [-\pi,\pi)$ are the amplitude and phase of the reflection coefficient of the $n$-th element on the RIS, respectively. We assume that the amplitudes of the reflection coefficients are known and unchangeable. Let $\boldsymbol{\phi}=[\boldsymbol{\phi}_1^T,\dots,\boldsymbol{\phi}_L^T]^T=[\beta_1 e^{j\theta_{1}},\ldots,\beta_n e^{j\theta_{n}},\ldots,\beta_N e^{j\theta_{N}}]^T$ be the gain vector of all the RISs. Then, the multiuser MIMO channel matrix $\boldsymbol{H} \in \mathbb{C}^{K \times M}$ can be represented as
\begin{equation} 
\boldsymbol{H}=\boldsymbol{H}_{r} \boldsymbol{\Phi} \boldsymbol{S} + \boldsymbol{H}_{d},
\end{equation}
where $\boldsymbol{\Phi} = \textrm{diag}(\boldsymbol{\phi})$. For the beauty of the subsequent formulas in the paper, let phase-shift vector $\boldsymbol{v}=[e^{j\theta_{1}},\ldots, e^{j\theta_{N}}]^H$. Note that $\boldsymbol{\Phi}=\boldsymbol{\Psi}\boldsymbol{\Lambda}$, where $\boldsymbol{\Psi}=\textrm{diag}\left(\boldsymbol{v}^{\star}\right)$ and $\boldsymbol{\Lambda}=\textrm{diag}\left([\beta_1,\ldots,\beta_N]\right)$. Then, (1) can be rewritten as
\begin{equation} 
\boldsymbol{H}=\boldsymbol{H}_{r} \boldsymbol{\Psi} \boldsymbol{G} + \boldsymbol{H}_{d},
\end{equation}
where $\boldsymbol{G}=\boldsymbol{\Lambda}\boldsymbol{S}$ is the modified baseband equivalent channel matrix of the BS-RIS link.

We assume that all the $M$ BS antennas share a common radio frequency chain to reduce the implementation cost. The radio frequency signal is generated as 
\begin{equation} 
z(t)=\sqrt{\frac{P}{M}} e^{j 2\pi f_c t},
\end{equation}
where $f_c$ is the carrier frequency and $P$ is the total transmission power at the BS. Then, $z(t)$ goes through the phase shifter at antenna $m$, yielding 
\begin{equation} 
s_m (t)=z(t)e^{j\alpha_m}=x_m e^{j 2\pi f_c t}, m=1,\dots,M,
\end{equation}
where $x_m=\sqrt{P/M} e^{j \alpha_m}$ and $\alpha_m \in [-\pi,\pi)$ is the phase-shift at the $m$-th antenna.  At time $t$, the received signal vector of all the $K$ users is given by
\begin{equation} 
\boldsymbol{y}(t)=\boldsymbol{H}\boldsymbol{s}(t),
\end{equation}
where $\boldsymbol{s}(t) = [s_1(t),s_2(t),\dots,s_M(t)]^T \in \mathbb{C}^{M \times 1}$ with the element $s_m(t)$ given by (4) and $\boldsymbol{y}(t) = [y_1(t),y_2(t),\dots,y_K(t)]^T \in \mathbb{C}^{K \times 1} $ with $y_k(t)$ being the received signal of user $k$. Note that in (5), the additive noise is ignored for energy harvesting, since it is too weak for the received signals \cite{zhang2017multi}. In addition, the channel is quasi-static, i.e., $\boldsymbol{H}$ remains unchanged within the channel coherence time.

Each user aims to harvest energy from the received radio frequency signal $y_k(t)$. The received power of user $k$ is given by
\begin{equation} 
Q_k=\frac{\eta}{T} \int_{0}^{T} \left|y_k(t)\right|^2 dt =\eta\left|\boldsymbol{h}_k^H\boldsymbol{x}\right|^2,
\end{equation}
where $\eta \in \left[0,1\right]$ is the energy conversion efficiency of the receivers, $T=1/f_c$ is the period of the carrier signal, $\boldsymbol{h}_{k}^H$ is the $k$-th row of the channel matrix $\boldsymbol{H}$, and $\boldsymbol{x}=[x_1,\dots,x_M]^T$.

\subsection{Problem Formulation}

In this subsection we introduce the two optimization problems studied in this paper.

\subsubsection{Sum Power Maximization (SPM)}

The objective of the first optimization problem is to maximize the total received power by jointly optimizing the transmit signal vector $\boldsymbol{x}$ and the phase-shift matrix $\boldsymbol{\Psi}$, subject to the constant-envelope constraints of the transmit signals $\left\{[\boldsymbol{x}]_m\right\}_{m=1}^{M}$ and the RIS phase-shift elements $\{[\boldsymbol{\Psi}]_{n,n}\}_{n=1}^{N}$. Note that the total received power can be expressed as $\eta\left\|(\boldsymbol{H}_{r} \boldsymbol{\Psi} \boldsymbol{G} + \boldsymbol{H}_{d})\boldsymbol{x}\right\|^2$. Since $\eta$ is a constant, it can be ignored in the sum power maximization problem, as shown below:
\begin{alignat}{2}
\textrm{(P1):} \quad \max_{\boldsymbol{\Psi}, \boldsymbol{x}} & \ Q = \left\|(\boldsymbol{H}_{r} \boldsymbol{\Psi} \boldsymbol{G} + \boldsymbol{H}_{d})\boldsymbol{x}\right\|^2 & \\
\mbox{s.t.}\quad
&[\boldsymbol{\Psi}]_{n,n}=e^{j\theta_{n}}, \theta_n \in [-\pi,\pi), n=1,\dots,N,  & \\
&[\boldsymbol{x}]_m=\sqrt{\frac{P}{M}}e^{j\alpha_m}, \alpha_m \in [-\pi,\pi), m=1,\dots,M.
\end{alignat}
Note that, since the power consumed by the transmitter is fixed, maximizing the total received power is identical with maximizing the energy efficiency. Due to the constant-envelope constraints and the non-concave objective funtion with respect to $\boldsymbol{\Psi}$ and $\boldsymbol{x}$, (P1) is non-convex and is in general difficult to solve exactly. however, if we fix $\boldsymbol{\Psi}$ (or $\boldsymbol{x}$), (P1) is a non-convex quadratically constrained quadratic program (QCQP) with respect to $\boldsymbol{x}$ (or $\boldsymbol{\Psi}$). This observation inspired us to leverage alternating optimization technique, which is an iterative procedure for maximizing the objective function by alternating maximizations over the individual subsets of the variables.

\subsubsection{Sum Power Maximization with minimum received power Constraints (SPMC)}
The sum power maximization optimization objective sacrifices the fairness among individual users to improve the total received power, especially when the user channel suffers from the near-far effect. In practical system design, fairness of wireless power transfer among users needs to be traded off with power efficiency.

To strike a balance between power efficiency and user fairness, we consider the optimization problem of maximizing the total received power with the minimum received power constraints, formulated as

\begin{alignat}{2}
\textrm{(P2):} \quad \max_{\boldsymbol{\Psi},\boldsymbol{x}}& \  
\sum_{k=1}^{K} \left(Q_k= \left|\boldsymbol{h}_{r,k}^{H} \boldsymbol{\Psi}\boldsymbol{G}\boldsymbol{x}+\boldsymbol{h}_{d,k}^{H}\boldsymbol{x} \right|^2\right) & \\
\mbox{s.t.}\quad
&[\boldsymbol{\Psi}]_{n,n}=e^{j\theta_{n}}, \theta_n \in [-\pi,\pi), n=1,\dots,N,\\
&[\boldsymbol{x}]_m=\sqrt{\frac{P}{M}}e^{j\alpha_m}, \alpha_m \!\in\! [-\pi,\pi), m=1,\dots,M,\\
&\left|\boldsymbol{h}_{r,k}^{H} \boldsymbol{\Psi}\boldsymbol{G}\boldsymbol{x}+\boldsymbol{h}_{d,k}^{H}\boldsymbol{x} \right|^2 \geq \frac{p_k}{\eta},k=1,\dots,K,
\end{alignat}
where $p_k$ refers to the minimum received power allowed for user $k$, $\boldsymbol{h}_{r,k}^{H}$ and $\boldsymbol{h}_{d,k}^{H}$ represent the $k$-th rows of matrix $\boldsymbol{H}_r$ and $\boldsymbol{H_d}$, respectively. Note that $\eta\left|\boldsymbol{h}_{r,k}^{H} \boldsymbol{\Psi}\boldsymbol{G}\boldsymbol{x}+\boldsymbol{h}_{d,k}^{H}\boldsymbol{x} \right|^2$ refers to the power received by user $k$. For the same reason, we solve this problem by leveraging alternating optimization technique.

\section{Sum Power Maximization}

In this section, we propose two suboptimal solutions based on alternating optimization for sum power maximization by leveraging the SDR and SCA techniques. The corresponding algorithms are called SPM-SDR and SPM-SCA.


\subsection{SPM-SDR Algorithm}

In this subsection, we apply the SDR technique in the two steps of alternating optimization separately to solve (P1).

\subsubsection{Optimization of $\boldsymbol{x}$ for a fixed $\boldsymbol{\Psi}$}
At this point, multiuser MIMO channel $\boldsymbol{H}$ is fixed because of the fixed phase-shift matrix $\boldsymbol{\Psi}$, we then focus on solving
\begin{alignat}{2}
\textrm{(P3):} \quad \max_{\boldsymbol{x}} \quad & \|\boldsymbol{H}\boldsymbol{x}\|^2 & \\
\mbox{s.t.}\quad
&\left|[\boldsymbol{x}]_m\right|=\sqrt{\frac{P}{M}}, m=1,\dots,M.
\end{alignat}

Note that $\|\boldsymbol{H}\boldsymbol{x}\|^2=\textrm{tr}(\boldsymbol{H}^{H}\boldsymbol{H}\boldsymbol{x}\boldsymbol{x}^{H})$. Define $\boldsymbol{X}=\boldsymbol{x}\boldsymbol{x}^{H}$ where $\boldsymbol{X} \succeq \boldsymbol{0}$ and rank$(\boldsymbol{X})=1$. Then (P3) can be equivalently rewritten as
\begin{alignat}{2}
\textrm{(P4):} \quad \max_{\boldsymbol{X}} \quad & \textrm{tr}(\boldsymbol{H}^{H}\boldsymbol{H}\boldsymbol{X}) & \\
\mbox{s.t.}\quad
&\boldsymbol{X} \succeq \boldsymbol{0};\
[\boldsymbol{X}]_{m,m}=\frac{P}{M}, m=1,\dots,M,\\
&\textrm{rank}(\boldsymbol{X}) =1.
\end{alignat}

By dropping the rank-one constraint, (P4) is a standard convex semidefinite programming (SDP) problem and hence can be solved by existing convex optimization solvers such as CVX. Note that, in general, the solution to the SDP problem is not a rank-one matrix.  Thus, additional steps are needed to construct a rank-one solution from this optimal higher-rank solution \cite{Luo2010Semidefinite, Wiesel2005Semidefinite}. Specifically, let the eigenvalue decomposition of $\boldsymbol{X}$ as $\boldsymbol{X}=\boldsymbol{U}\boldsymbol{\Sigma} \boldsymbol{U}^H$, where $\boldsymbol{U}=\left[\boldsymbol{e}_1,\dots,\boldsymbol{e}_M\right] \in \mathbb{C}^{M \times M}$ is a unitary matrix with each column is an eigenvector of $\boldsymbol{X}$, and $\boldsymbol{\Sigma}=\textrm{diag}\left([\lambda_1,\dots,\lambda_M]^T\right) \in \mathbb{C}^{M \times M}$ is a diagonal matrix with diagonal elements are the eigenvalues of $\boldsymbol{X}$. Then, generate $s$ complex vectors $\{\boldsymbol{u}_i =\textrm{exp}(j\textrm{arg}(\boldsymbol{U}\boldsymbol{\Sigma}^{1/2}\bar{\boldsymbol{u}}_i))\}_{i=1}^s$, where $\bar{\boldsymbol{u}}_i$ subject to $\mathcal{C}\mathcal{N}(\boldsymbol{0},\boldsymbol{I}_M)$, $\forall i$. Finally, we construct a vector $\hat{\boldsymbol{x}}=\sqrt{P/M}\hat{\boldsymbol{u}}$ as a suboptimal solution to (P3), where $\hat{\boldsymbol{u}}=\mathop{\arg\max}_{i=1,\dots,s}\left\|\boldsymbol{H}\boldsymbol{u}_i\right\|^2$.

\subsubsection{Optimization of $\boldsymbol{\Psi}$ for a fixed $\boldsymbol{x}$}
With $\boldsymbol{x}$ fixed, (P1) can be written as follows:
\begin{alignat}{2}
\textrm{(P5):} \quad \max_{\boldsymbol{\Psi}} \quad & 
\sum_{k=1}^{K}\left|\boldsymbol{h}_{r,k}^{H} \boldsymbol{\Psi}\boldsymbol{G}\boldsymbol{x}+\boldsymbol{h}_{d,k}^{H}\boldsymbol{x} \right|^2 & \\
\mbox{s.t.}\quad
&\left|[\boldsymbol{\Psi}]_{n,n}\right|=1, n=1,\dots,N.
\end{alignat}

Recall that $\boldsymbol{\Psi}=\textrm{diag}\left(\boldsymbol{v}^{\star}\right)$. Let $\boldsymbol{c}_k=\textrm{diag}(\boldsymbol{h}_{r,k}^{\star})\boldsymbol{G}\boldsymbol{x}$ and $a_k=\boldsymbol{h}_{d,k}^H\boldsymbol{x}$, we have $\left|\boldsymbol{h}_{r,k}^{H} \boldsymbol{\Psi}\boldsymbol{G}\boldsymbol{x}+\boldsymbol{h}_{d,k}^{H}\boldsymbol{x} \right|^2 = \left|\boldsymbol{v}^H\boldsymbol{c}_k+a_k\right|^2$. Note that $a_k a_k^{\star}$ is independent of $\boldsymbol{v}$, thus (P5) is equivalent to
\begin{alignat}{2}
\textrm{(P6):} \quad \max_{\boldsymbol{v}} \quad & 
\sum_{k=1}^{K} \boldsymbol{v}^H\boldsymbol{c}_k\boldsymbol{c}_k^H \boldsymbol{v}+a_k^{\star}\boldsymbol{v}^H\boldsymbol{c}_k+a_k\boldsymbol{c}_k^H\boldsymbol{v}& \\
\mbox{s.t.}\quad
&\left|[\boldsymbol{v}]_n\right|=1, n=1,\dots,N.&
\end{alignat}
By introducing an auxiliary variable $t$, (P6) can be recast as
\begin{alignat}{2}
\textrm{(P7):} \quad \max_{\bar{\boldsymbol{v}}} \quad & 
\bar{\boldsymbol{v}}^H\boldsymbol{R}\bar{\boldsymbol{v}}& \\
\mbox{s.t.}\quad
&\left|[\bar{\boldsymbol{v}}]_n\right|=1, n=1,\dots,N+1,&
\end{alignat}
where 
\begin{equation} 
\boldsymbol{R}=\sum_{k=1}^K
\begin{bmatrix}
\boldsymbol{c}_k\boldsymbol{c}_k^H& a_k^{\star}\boldsymbol{c}_k\\
a_k\boldsymbol{c}_k^H & 0
\end{bmatrix}
\ \textrm{and} \
\bar{\boldsymbol{v}}=
\begin{bmatrix}
t\boldsymbol{v}\\
t
\end{bmatrix}.
\end{equation}

Note that $\bar{\boldsymbol{v}}^H\boldsymbol{R}\bar{\boldsymbol{v}} = \textrm{tr}(\boldsymbol{R}\bar{\boldsymbol{v}}\bar{\boldsymbol{v}}^H)$. Define $\boldsymbol{V}=\bar{\boldsymbol{v}}\bar{\boldsymbol{v}}^H$ where $\boldsymbol{V} \succeq \boldsymbol{0}$ and rank$(\boldsymbol{V})=1$. Then, (P7) is equivalent to
\begin{alignat}{2}
\textrm{(P8):} \quad \max_{\boldsymbol{V}} \quad & 
\textrm{tr}(\boldsymbol{R}\boldsymbol{V})& \\
\mbox{s.t.}\quad
&\boldsymbol{V} \succeq \boldsymbol{0};\ 
[\boldsymbol{V}]_{n,n}=1, n=1,\dots,N+1,\\
&\textrm{rank}(\boldsymbol{V})=1,
\end{alignat}

Again, by dropping the rank-one constraint, (P8) is an SDP problem and can be solved by CVX. Through the randomization approach mentioned before, we can get a suboptimal solution of (P7), denoted by $\hat{\bar{\boldsymbol{v}}}$. Then the suboptimal solution of (P5) is given by $\hat{\boldsymbol{\Psi}}=\textrm{diag}\left(\textrm{exp}\left(j\textrm{arg}\left(\right.\right.\right.$ $\left.\left.\left.\left([\hat{\bar{\boldsymbol{v}}}]_{(1:N)}/[\hat{\bar{\boldsymbol{v}}}]_{N+1}\right)^{\star}\right)\right)\right)$.

\subsubsection{Overall Algorithm}

We summarize the proposed SDR based alternating optimization method in Algorithm 1. Specifically, the algorithm starts with certain feasible values of $\boldsymbol{x}^{(0)}$ and $\boldsymbol{\Psi}^{(0)}$. Next, given $\{\boldsymbol{x}^{(i)},\boldsymbol{\Psi}^{(i)}\}$ in the $i$-th iteration, we first update $\boldsymbol{x}$ to obtain solution $\{\boldsymbol{x}^{(i+1)},\boldsymbol{\Psi}^{(i)}\}$ as described in the first part of this subsection, and then on this basis, update $\boldsymbol{\Psi}$ to obtain solution $\{\boldsymbol{x}^{(i+1)},\boldsymbol{\Psi}^{(i+1)}\}$ as described in the second part of this subsection. Keep iterating until the termination condition is satisfied. 

The worst-case complexities of solving (P4) and (P8) using the SDR algorithm are $\mathcal{O}(M^{6.5})$ and $\mathcal{O}(N^{6.5})$, respectively. Thus the complexity of the SPM-SDR algorithm is $\mathcal{O}(I_{sdr} (M^{6.5} + N^{6.5}))$, where $I_{sdr}$ denotes the maximum iteration number. Clearly, the SDR based algorithm developed in this subsection is computationally involving. To reduce complexity, we propose the SCA based algorithm in next subsection.

\begin{algorithm}[h]
	\caption{SPM-SDR Algorithm}
	\label{alg::conjugateGradient}
	\begin{algorithmic}[1]
		\Require
		$\boldsymbol{H}_d, \boldsymbol{H}_r, \boldsymbol{G}, P$.
		\Ensure
		solution $\{\boldsymbol{x}^{op}$, $\boldsymbol{\Psi}^{op}\}$.
		\State Initialize $\boldsymbol{x}^{(0)}$ and $\boldsymbol{\Psi}^{(0)}$ to feasible values, initialize iteration number $i=0$ and threshold $\epsilon > 0$.
		\Repeat
		\State With $\boldsymbol{\Psi}^{(i)}$ fixed, the SDR technique is used to solve (P3) to obtain $\boldsymbol{x}^{(i+1)}$.
		\State With $\boldsymbol{x}^{(i+1)}$ fixed, the SDR technique is used to solve (P5) to obtain $\boldsymbol{\Psi}^{(i+1)}$.
		\State Update $i=i+1$.
		\Until{the fractional increase of the objective value is below the threshold $\epsilon$ or the maximum number of iterations is reached.}
	\end{algorithmic}
\end{algorithm}

\subsection{SPM-SCA Algorithm}

We first transform the form of the problem.
Let $\boldsymbol{A}_k=\textrm{diag}(\boldsymbol{h}_{r,k}^{\star})\boldsymbol{G}$, we have $\left|\boldsymbol{h}_{r,k}^{H} \boldsymbol{\Psi}\boldsymbol{G}\boldsymbol{x}+\boldsymbol{h}_{d,k}^{H}\boldsymbol{x} \right|^2 = \left|\boldsymbol{v}^H\boldsymbol{A}_k\boldsymbol{x}+\boldsymbol{h}_{d,k}^H\boldsymbol{x}\right|^2$. Then (P1) can be rewritten as
\begin{alignat}{2}
\textrm{(P9):} \quad \max_{\boldsymbol{v},\boldsymbol{x}} \quad & 
Q=\sum_{k=1}^{K}\left|\boldsymbol{v}^H\boldsymbol{A}_k\boldsymbol{x}+\boldsymbol{h}_{d,k}^H\boldsymbol{x}\right|^2 & \\
\mbox{s.t.}\quad
&\left|[\boldsymbol{v}]_n\right|=1, n=1,\dots,N,\\
&\left|[\boldsymbol{x}]_m\right|=\sqrt{\frac{P}{M}}, m=1,\dots,M.
\end{alignat}

In this subsection, we solve (P9) by leveraging SCA technique. SCA is an iterative procedure: at each step, the objective function of the original problem is replaced by an approximation around a feasible point, and then the resulting approximation problem is solved to obtain the approximation point for the next iteration.

\subsubsection{Optimization of $\boldsymbol{x}$ for a fixed $\boldsymbol{v}$}

Given a fixed $\boldsymbol{v}$ and a feasible point $\hat{\boldsymbol{x}}$, we obtain from convexity that
\begin{equation} 
\begin{split}
\left|\boldsymbol{v}^H\boldsymbol{A}_k\boldsymbol{x}+\boldsymbol{h}_{d,k}^H\boldsymbol{x}\right|^2 \geq 2\textrm{Re}\{\boldsymbol{x}^H\boldsymbol{B}_k(\boldsymbol{v})\hat{\boldsymbol{x}}\}-
\hat{\boldsymbol{x}}^H \boldsymbol{B}_k(\boldsymbol{v})\hat{\boldsymbol{x}},
\end{split}
\end{equation}
where $\boldsymbol{B}_k(\boldsymbol{v})=(\boldsymbol{v}^H\boldsymbol{A}_k+\boldsymbol{h}_{d,k}^H)^H(\boldsymbol{v}^H\boldsymbol{A}_k+\boldsymbol{h}_{d,k}^H)$. Let $\boldsymbol{B}(\boldsymbol{v})=\sum\limits_{k=1}^K\boldsymbol{B}_k(\boldsymbol{v})=\boldsymbol{H}^H\boldsymbol{H}$, we have
\begin{equation} 
\begin{split}
Q \geq 2\textrm{Re}\{\boldsymbol{x}^H\boldsymbol{B}(\boldsymbol{v})\hat{\boldsymbol{x}}\}-
\hat{\boldsymbol{x}}^H \boldsymbol{B}(\boldsymbol{v})\hat{\boldsymbol{x}}.
\end{split}
\end{equation}
where the equality holds at point $\boldsymbol{x}=\hat{\boldsymbol{x}}$. (33) gives an approximation of the objective function of (P9) when $\boldsymbol{v}$ is fixed. Next, we maximize this approximation function subject to the constraint (31). It is easy to see that the maximizer is given by
\begin{equation} 
\begin{split}
\boldsymbol{x}^{op}=\sqrt{\frac{P}{M}} \textrm{exp}\left(j\textrm{arg}\left(\boldsymbol{B}(\boldsymbol{v})\hat{\boldsymbol{x}}\right)\right).
\end{split}
\end{equation}

\subsubsection{Optimization of $\boldsymbol{v}$ for a fixed $\boldsymbol{x}$}
Similarly, for a given feasible point $\hat{\boldsymbol{v}}$, we have
\begin{equation} 
\begin{split}
\left|\boldsymbol{v}^H\boldsymbol{A}_k\boldsymbol{x}+\boldsymbol{h}_{d,k}^H\boldsymbol{x}\right|^2 \geq 2\textrm{Re}\{\boldsymbol{v}^H\boldsymbol{C}_k(\boldsymbol{x},\hat{\boldsymbol{v}})\}\\
-(\hat{\boldsymbol{v}}^H\boldsymbol{A}_k\boldsymbol{x}\boldsymbol{x}^H\boldsymbol{A}_k^H\hat{\boldsymbol{v}}-\boldsymbol{h}_{d,k}^H\boldsymbol{x}\boldsymbol{x}^H\boldsymbol{h}_{d,k}),
\end{split}
\end{equation}
where $\boldsymbol{C}_k(\boldsymbol{x},\hat{\boldsymbol{v}})=\boldsymbol{C}_k^1(\boldsymbol{x})\hat{\boldsymbol{v}}+\boldsymbol{C}_k^2(\boldsymbol{x})$, $\boldsymbol{C}_k^1(\boldsymbol{x})=\boldsymbol{A}_k\boldsymbol{x}\boldsymbol{x}^H\boldsymbol{A}_k^H$, and $\boldsymbol{C}_k^2(\boldsymbol{x})=\boldsymbol{A}_k\boldsymbol{x}\boldsymbol{x}^H\boldsymbol{h}_{d,k}$. 

Let $\boldsymbol{C}^1(\boldsymbol{x})=\sum\limits_{k=1}^K\boldsymbol{C}_k^1(\boldsymbol{x})$, $\boldsymbol{C}^2(\boldsymbol{x})=\sum\limits_{k=1}^K\boldsymbol{C}_k^2(\boldsymbol{x})$, and $\boldsymbol{C}(\boldsymbol{x},\hat{\boldsymbol{v}})=\boldsymbol{C}^1(\boldsymbol{x})\hat{\boldsymbol{v}}+\boldsymbol{C}^2(\boldsymbol{x})$. Then
\begin{equation} 
\begin{split}
Q \!\geq\! 2\textrm{Re}\{\!\boldsymbol{v}^H\boldsymbol{C}\!(\boldsymbol{x},\hat{\boldsymbol{v}})\!\}\!-\!
\sum_{k=1}^{K}\!(\hat{\boldsymbol{v}}^H\boldsymbol{A}_k\boldsymbol{x}\boldsymbol{x}^H\boldsymbol{A}_k^H\hat{\boldsymbol{v}}\!-\!
\boldsymbol{h}_{d,k}^H\boldsymbol{x}\boldsymbol{x}^H\boldsymbol{h}_{d,k}),
\end{split}
\end{equation}
where the equality holds at point $\boldsymbol{v}=\hat{\boldsymbol{v}}$. With the constraint (30), the point that maximizes this lower bound is given by 
\begin{equation} 
\begin{split}
\boldsymbol{v}^{op}=\textrm{exp}\left(j\textrm{arg}\left(\boldsymbol{C}^1(\boldsymbol{x})\hat{\boldsymbol{v}}+\boldsymbol{C}^2(\boldsymbol{x})\right)\right).
\end{split}
\end{equation}

\subsubsection{Overall Algorithm}

We summarize the proposed SCA based alternating optimization method in Algorithm 2. Specifically, the algorithm starts with certain feasible values of $\boldsymbol{x}^{(0)}$ and $\boldsymbol{v}^{(0)}$. Then, iteratively update $\boldsymbol{x}$ and $\boldsymbol{v}$ using the update formulas (34) and (37) until the termination condition is satisfied. In the end, we obtain $\boldsymbol{x}^{op}=\boldsymbol{x}^{(\bold i)}$ and $\boldsymbol{\Psi}^{op}=\textrm{diag}\left((\boldsymbol{v}^{(\bold i)})^{\star}\right)$, where $\bold i$ is the number of iterations when the algorithm terminates. This algorithm is convergent since the alternating optimization and SCA techniques ensure the objective value of (P1) monotonically non-decreasing in the iterative process.

It is not difficult to see that in each iteration step, the complexities of updating $\boldsymbol{x}$ and $\boldsymbol{v}$ are $\mathcal{O}(M^2)$ and $\mathcal{O}(N^2)$, respectively. Thus the complexity of the SPM-SCA algorithm is $\mathcal{O}(I_{sca}(M^2 + N^2))$, where $I_{sca}$ denotes the number of iterations required for the algorithm to converge. Recall that the SDR based algorithm proposed in the previous subsection requires a much higher complexity of $\mathcal{O}(I_{sdr} (M^{6.5} + N^{6.5}))$.

\begin{algorithm}[h]
	\caption{SPM-SCA Algorithm}
	\label{alg::conjugateGradient}
	\begin{algorithmic}[1]
		\Require
		$\boldsymbol{H}_d, \boldsymbol{H}_r, \boldsymbol{G}, P$.
		\Ensure
		solution $\{\boldsymbol{x}^{op}$, $\boldsymbol{\Psi}^{op}=\textrm{diag}\left((\boldsymbol{v}^{op})^{\star}\right)\}$.
		\State Initialize $\boldsymbol{x}^{(0)}$ and $\boldsymbol{v}^{(0)}$ to feasible values, initialize iteration number $i=0$ and threshold $\epsilon > 0$.
		\Repeat
		\State $\boldsymbol{x}^{(i+1)}=\sqrt{P/M} \textrm{exp}\left(j\textrm{arg}\left(\boldsymbol{B}(\boldsymbol{v}^{(i)})\boldsymbol{x}^{(i)}\right)\right)$.
		\State $\boldsymbol{v}^{(i+1)}=\textrm{exp}\left(j\textrm{arg}\left(\boldsymbol{C}^1(\boldsymbol{x}^{(i+1)})\boldsymbol{v}^{(i)}+\boldsymbol{C}^2(\boldsymbol{x}^{(i+1)})\right)\right)$.
		\State Update $i=i+1$.
		\Until{the fractional increase of the objective value is below the threshold $\epsilon$ or the maximum number of iterations is reached.}
	\end{algorithmic}
\end{algorithm}

\section{Sum Power Maximization with Minimum Received Power Constraints}

In this section, we propose a suboptimal solution to (P2) by dividing (P2) into two subproblems and solving the two subproblems iteratively by leveraging the SCA and ADMM techniques. The proposed algorithm is called SPMC-SCA-ADMM.



\subsection{Optimization of $\boldsymbol{x}$ for a fixed $\boldsymbol{\Psi}$}

With $\boldsymbol{\Psi}$ fixed, (P2) can be written as follows:
\begin{alignat}{2}
\textrm{(P10):} \quad \max_{\boldsymbol{x}} \quad & 
\|\boldsymbol{H}\boldsymbol{x}\|^2 & \\
\mbox{s.t.}\quad
&\left|\boldsymbol{h}_k^H\boldsymbol{x}\right|^2 \geq \frac{p_k}{\eta},k=1,\dots,K,\\
&\left|[\boldsymbol{x}]_m\right|=\sqrt{\frac{P}{M}}, m=1,\dots,M.
\end{alignat}
Recall that $\boldsymbol{h}_k^H$ denote the $k$-th row of the channel matrix $\boldsymbol{H}$. For such a non-convex QCQP, the traditional method to tackle this problem is through the SDR method: first solve the corresponding SDP problem, then a approximate solution of the original problem is obtained by applying the SDR randomization method \cite{Luo2010Semidefinite, Wiesel2005Semidefinite}. However, for (P10), it is difficult to obtain a feasible point from the randomization method \cite{huang2016consensus}. More specifically, it is difficult to map a randomly selected rank-one candidate (from the SDR solution) to the feasible region determined jointly by (39) and (40). Alternatively, to tackle this problem, we first use the SCA framework to transform the problem form, and then use the ADMM algorithm to solve it.

To apply the SCA method, we need to find a suitable lower bound of $\|\boldsymbol{H}\boldsymbol{x}\|^2$. To do so, we expand $\|\boldsymbol{H}\boldsymbol{x}\|^2$ at a feasible point $\hat{\boldsymbol{x}}$ to obtain a linear lower bound given by
\begin{equation}
\|\boldsymbol{H}\boldsymbol{x}\|^2 \geq 2 \text{Re}\left\{\hat{\boldsymbol{x}}^H\boldsymbol{H}^H\boldsymbol{H}\boldsymbol{x}\right\}-\hat{\boldsymbol{x}}^H\boldsymbol{H}^H\boldsymbol{H}\hat{\boldsymbol{x}},
\end{equation}
where the equality holds at point $\boldsymbol{x}=\hat{\boldsymbol{x}}$. Next, we use the ADMM algorithm to maximize this lower bound under the constraints of (39) and (40). The corresponding optimization problem is as follows
\begin{alignat}{2}
\textrm{(P11):} \quad \max_{\boldsymbol{x}} \quad & \text{Re}\left\{\hat{\boldsymbol{x}}^H\boldsymbol{H}^H\boldsymbol{H}\boldsymbol{x}\right\}& \\
\mbox{s.t.}\quad
&\textrm{(39), (40)}.
\end{alignat}
(P11) can be written in the following form
\begin{alignat}{2}
\textrm{(P12):} \quad \min_{\boldsymbol{x},\left\{\boldsymbol{e}_k \right\}_{k=1}^{K}} &
\text{Re}\left\{-\hat{\boldsymbol{x}}^H\boldsymbol{H}^H\boldsymbol{H}\boldsymbol{x}\right\} \\
\mbox{s.t.}\quad
&\left|\boldsymbol{h}_k^H \boldsymbol{e}_k\right|^2 \geq \frac{p_k}{\eta}, k=1,\dots,K,\\
&\left|\left[\boldsymbol{x}\right]_m\right|=\sqrt{\frac{P}{M}}, m=1,\dots,M,\\
&\boldsymbol{e}_k=\boldsymbol{x}, k=1,\dots,K.
\end{alignat}

Define the feasible region of constraint (45) as $\mathcal{G}$, whose indicator function is given by 
\begin{equation}
\mathbb{I}_{\mathcal{G}}\left(\{\boldsymbol{e}_k\}_{k=1}^{K}\right)=\left\{\begin{array}{ll}
0, & \text { if } \{\boldsymbol{e}_k\}_{k=1}^{K} \in \mathcal{G}, \\
+\infty, & \text { otherwise }.
\end{array}\right.
\end{equation}
Similarly, define the feasible region of constraint (46) as $\mathcal{H}$, and its indicator function as
\begin{equation}
\mathbb{I}_{\mathcal{H}}(\boldsymbol{x})=\left\{\begin{array}{ll}
0, & \text { if } \boldsymbol{x} \in \mathcal{H}, \\
+\infty, & \text { otherwise }.
\end{array}\right.
\end{equation}
Thus, we obtain the equivalent ADMM form of (P12) as
\begin{alignat}{2}
\textrm{(P13):}  \min_{\boldsymbol{x},\left\{\boldsymbol{e}_k \right\}_{k=1}^{K}}  
&\text{Re}\left\{-\hat{\boldsymbol{x}}^H\boldsymbol{H}^H\boldsymbol{H}\boldsymbol{x}\right\} + \mathbb{I}_{\mathcal{G}}\left(\{\boldsymbol{e}_k\}_{k=1}^{K}\right) + \mathbb{I}_{\mathcal{H}}(\boldsymbol{x})\\
\mbox{s.t.}\quad 
&\textrm{(47)}.
\end{alignat}
Then, the augmented Lagrangian of (P13) can be formulated as
\begin{equation}
\begin{split}
\mathcal{L}_{\rho}\left(\boldsymbol{x},\left\{\boldsymbol{e}_k \right\}_{k=1}^{K}, \left\{\boldsymbol{u}_k \right\}_{k=1}^{K}\right)
=\text{Re}\left\{-\hat{\boldsymbol{x}}^H\boldsymbol{H}^H\boldsymbol{H}\boldsymbol{x}\right\} \\ +\mathbb{I}_{\mathcal{G}}\left(\{\boldsymbol{e}_k\}_{k=1}^{K}\right) + \mathbb{I}_{\mathcal{H}}(\boldsymbol{x}) 
+ \rho \sum_{k=1}^K \left\|\boldsymbol{e}_k-\boldsymbol{x}+\boldsymbol{u}_k\right\|^2,
\end{split}
\end{equation}
where $\rho \textgreater 0$ is the penalty parameter, $\{\boldsymbol{u}_k\}_{k=1}^K$ are the scaled dual variables. Applying the ADMM method, we update the global variable $\boldsymbol{x}$, the local variables $\left\{\boldsymbol{e}_k\right\}_{k=1}^{K}$ and the scaled dual variables $\left\{\boldsymbol{u}_k \right\}_{k=1}^{K}$ alternatively. It is worth mentioning that the element variables of each of the above three sets can be updated synchronously, while the variables of different sets need to be updated sequentially, which leaves room for parallel computation. 

In the $i$-th iteration, given $\boldsymbol{x}^{(i)},\{\boldsymbol{e}_k^{(i)}\}_{k=1}^{K}$ and $\{\boldsymbol{u}_k^{(i)}\}_{k=1}^{K}$, we update each of the above variables as follows.

\paragraph{Update $\boldsymbol{x}$}
The subproblem for updating the global variables $\boldsymbol{x}$ is expressed as
\begin{equation}
\begin{split}
\boldsymbol{x}^{(i+1)}=&\mathop{\arg\min}_{\boldsymbol{x}} \ \mathcal{L}_{\rho}\left(\boldsymbol{x},\left\{\boldsymbol{e}_k^{(i)}\right\}_{k=1}^{K},\left\{\boldsymbol{u}_k^{(i)}\right\}_{k=1}^{K}\right),\\
=&\mathop{\arg\min}_{\boldsymbol{x}} \
\mathbb{I}_{\mathcal{H}}(\boldsymbol{x})+ K\rho\boldsymbol{x}^H\boldsymbol{x}\\
&-\text{Re}\left\{\left(\hat{\boldsymbol{x}}^H\boldsymbol{H}^H\boldsymbol{H} +2\rho \sum_{k=1}^K\left(\boldsymbol{u}^{(i)}_k+\boldsymbol{e}^{(i)}_k\right)^H\right)\boldsymbol{x}\right\}.
\end{split}
\end{equation}
Since $\boldsymbol{x}^H\boldsymbol{x}$ is a constant under the constraint (46), it is easy to see that the optimum is given by
\begin{equation}
\begin{split}
\boldsymbol{x}^{(i+1)}
&\!=\!\sqrt{\frac{P}{M}} \textrm{exp}\left(j\textrm{arg}\left(\boldsymbol{H}^H\boldsymbol{H}\hat{\boldsymbol{x}}\!+\!2\rho\sum_{k=1}^{K}\left(\boldsymbol{u}^{(i)}_k\!+\!\boldsymbol{e}^{(i)}_k\right)\right)\right).
\end{split}
\end{equation}

\paragraph{Update $\left\{\boldsymbol{e}_k\right\}_{k=1}^{K}$}
The subproblem for updating the local variables $\left\{\boldsymbol{e}_k\right\}_{k=1}^{K}$ is expressed as
\begin{equation}
\begin{split}
\left\{\boldsymbol{e}^{(i+1)}_k\right\}_{k=1}^{K}&=\mathop{\arg\min}_{\left\{\boldsymbol{e}_k\right\}_{k=1}^{K}} \ \mathcal{L}_{\rho}\left(\boldsymbol{x}^{(i+1)},\left\{\boldsymbol{e}_k\right\}_{k=1}^{K},\left\{\boldsymbol{u}_k^{(i)}\right\}_{k=1}^{K}\right),
\end{split}
\end{equation}
which can be rewritten as
\begin{alignat}{2}
\textrm{(P14):} \quad \min_{\left\{\boldsymbol{e}_k \right\}_{k=1}^{K}}\quad &
\sum_{k=1}^K \left\|\boldsymbol{e}_k-\boldsymbol{x}^{(i+1)}+\boldsymbol{u}^{(i)}_k\right\|^2 \\
\mbox{s.t.}\quad
&\left|\boldsymbol{h}_k^H \boldsymbol{e}_k\right|^2 \geq \frac{p_k}{\eta}, k=1,\dots,K.
\end{alignat}
See Appendix A for the derivation of the optimal solution to this problem. Here we directly give the optimal solution as follows: For $k=1,\dots,K$,
\begin{equation}
\begin{cases}
\boldsymbol{e}_k^{(i+1)}=\boldsymbol{x}^{(i+1)}-\boldsymbol{u}_k^{(i)};& \textrm{if} \left|\boldsymbol{h}_k^H\left(\boldsymbol{x}^{(i+1)}-\boldsymbol{u}_k^{(i)}\right)\right|^2 \geq \frac{p_k}{\eta},\\
\boldsymbol{\Gamma_k}\left(\boldsymbol{x}^{(i+1)},\boldsymbol{u}_k^{(i)}\right);& \textrm{otherwise},
\end{cases}
\end{equation}
where
\begin{equation}
\begin{split}
\boldsymbol{\Gamma_k}&\left(\boldsymbol{x}^{(i+1)},\boldsymbol{u}_k^{(i)}\right)=\boldsymbol{x}^{(i+1)}-\boldsymbol{u}_k^{(i)}\\
+&\frac{\sqrt{\frac{p_k}{\eta}}-\left|\boldsymbol{h}_k^H \left(\boldsymbol{x}^{(i+1)}-\boldsymbol{u}_k^{(i)}\right)\right|}{\left\|\boldsymbol{h}_k\right\|^2  \left|\boldsymbol{h}_k^H \left(\boldsymbol{x}^{(i+1)}-\boldsymbol{u}_k^{(i)}\right)\right|}\boldsymbol{h}_k \boldsymbol{h}_k^H  \left(\boldsymbol{x}^{(i+1)}-\boldsymbol{u}_k^{(i)}\right).
\end{split}
\end{equation}

\paragraph{Update  $\left\{\boldsymbol{u}_k\right\}_{k=1}^{K}$}
According to the ADMM method, the update formulas for scaled dual variables $\left\{\boldsymbol{u}_k\right\}_{k=1}^{K}$ are shown below:
\begin{equation}
\boldsymbol{u}_k^{(i+1)}=\boldsymbol{u}_k^{(i)}+\boldsymbol{e}_k^{(i+1)}-\boldsymbol{x}^{(i+1)}, k=1,\dots,K.
\end{equation}

Iterating using (54), (58) and (60), we finally obtain a suboptimal solution of (P11).

\subsection{Optimization of $\boldsymbol{\Psi}$ for a fixed $\boldsymbol{x}$}

With $\boldsymbol{x}$ fixed, we can rewrite (P2) as
\begin{alignat}{2}
\textrm{(P15):} \quad \max_{\boldsymbol{v}} \quad & 
\sum_{k=1}^{K} \left|\boldsymbol{v}^H\boldsymbol{c}_k+a_k\right|^2& \\
\mbox{s.t.}\quad
&\left|[\boldsymbol{v}]_n\right|=1, n=1,\dots,N,\\
&\left|\boldsymbol{v}^H\boldsymbol{c}_k+a_k\right|^2 \geq \frac{p_k}{\eta},k=1,\dots,K,
\end{alignat}
recall that $\boldsymbol{v}=[e^{j\theta_{1}},\ldots, e^{j\theta_{N}}]^H$, $\boldsymbol{c}_k=\textrm{diag}(\boldsymbol{h}_{r,k}^{\star})\boldsymbol{G}\boldsymbol{x}$ and $a_k=\boldsymbol{h}_{d,k}^H\boldsymbol{x}$. By introducing an auxiliary variable $t$, (P15) can be equivalently written as
\begin{alignat}{2}
\textrm{(P16):} \quad \max_{\boldsymbol{b}} \quad & 
\boldsymbol{b}^H\boldsymbol{L}\boldsymbol{b}& \\
\mbox{s.t.}\quad
&\left|\boldsymbol{l}_k^H\boldsymbol{b}\right|^2 \geq \frac{p_k}{\eta},k=1,\dots,K,\\
&\left|\left[\boldsymbol{b}\right]_n\right|=1, n=1,\dots,N+1,
\end{alignat}
where 
\begin{equation} 
\boldsymbol{l}_k=
\begin{bmatrix}
\boldsymbol{c}_k\\
a_k
\end{bmatrix},
\quad
\boldsymbol{L}=\sum_{k=1}^K \boldsymbol{l}_k \boldsymbol{l}_k^H
\ \textrm{and}\
\boldsymbol{b}=
\begin{bmatrix}
t\boldsymbol{v}\\
t
\end{bmatrix}.
\end{equation}

Again, we expand $\boldsymbol{b}^H\boldsymbol{L}\boldsymbol{b}$ at feasible point $\hat{\boldsymbol{b}}$ to obtain a linear lower bound of it as follow
\begin{equation}
\boldsymbol{b}^H\boldsymbol{L}\boldsymbol{b} \geq 2 \text{Re}\left\{\hat{\boldsymbol{b}}^H\boldsymbol{L}\boldsymbol{b}\right\}-\hat{\boldsymbol{b}}^H\boldsymbol{L}\hat{\boldsymbol{b}},
\end{equation}
where the equality holds at point $\boldsymbol{b}=\hat{\boldsymbol{b}}$. Next, we use the ADMM algorithm to maximize this lower bound under the constraints of (65) and (66). The corresponding optimization problem is as follows
\begin{alignat}{2}
\textrm{(P17):} \quad \max_{\boldsymbol{b}} \quad & 
\text{Re}\left\{\hat{\boldsymbol{b}}^H\boldsymbol{L}\boldsymbol{b}\right\}& \\
\mbox{s.t.}\quad
&\textrm{(65), (66)}.
\end{alignat}
Since (P17) has the same form as (P11), it can be solved by using the method proposed in the previous part. In the following, we first explain the meaning of the symbols, and then directly give the corresponding update formulas.

Let $\bar{\rho} \textgreater 0$ denote the penalty parameter, $\boldsymbol{b}$ denote the corresponding global variables, $\{\bar{\boldsymbol{e}}_k\}_{k=1}^{K}$ denote the corresponding local variables, and $\{\bar{\boldsymbol{u}}_k\}_{k=1}^{K}$ denote the corresponding scaled dual variables. We omit the derivation steps and directly give their update formulas as follows.

\paragraph{Update $\boldsymbol{b}$}
\begin{equation}
\begin{split}
\boldsymbol{b}^{(i+1)}
&=\textrm{exp}\left(j\textrm{arg}\left(\boldsymbol{L}\hat{\boldsymbol{b}}+2\bar{\rho}\sum_{k=1}^{K}\left(\bar{\boldsymbol{u}}^{(i)}_k+\bar{\boldsymbol{e}}^{(i)}_k\right)\right)\right).
\end{split}
\end{equation}

\paragraph{Update $\{\bar{\boldsymbol{e}}_k\}_{k=1}^{K}$}
For $k=1,\dots,K$, let 
\begin{equation}
\begin{split}
\boldsymbol{\Gamma_k}&\left(\boldsymbol{b}^{(i+1)},\bar{\boldsymbol{u}}_k^{(i)}\right)=\boldsymbol{b}^{(i+1)}-\bar{\boldsymbol{u}}_k^{(i)}\\
+&\frac{\sqrt{\frac{p_k}{\eta}}-\left|\boldsymbol{l}_k^H \left(\boldsymbol{b}^{(i+1)}-\bar{\boldsymbol{u}}_k^{(i)}\right)\right|}{\left\|\boldsymbol{l}_k\right\|^2  \left|\boldsymbol{l}_k^H \left(\boldsymbol{b}^{(i+1)}-\bar{\boldsymbol{u}}_k^{(i)}\right)\right|}\boldsymbol{l}_k \boldsymbol{l}_k^H \left(\boldsymbol{b}^{(i+1)}-\bar{\boldsymbol{u}}_k^{(i)}\right),
\end{split}
\end{equation}
we have
\begin{equation}
\bar{\boldsymbol{e}}_k^{(i+1)}=
\begin{cases}
\boldsymbol{b}^{(i+1)}-\bar{\boldsymbol{u}}_k^{(i)};& \textrm{if} \left|\boldsymbol{l}_k^H\left(\boldsymbol{b}^{(i+1)}-\bar{\boldsymbol{u}}_k^{(i)}\right)\right|^2 \geq \frac{p_k}{\eta},\\
\boldsymbol{\Gamma_k}\left(\boldsymbol{b}^{(i+1)},\bar{\boldsymbol{u}}_k^{(i)}\right) ;& \textrm{otherwise}.
\end{cases}
\end{equation}
\paragraph{Update $\{\bar{\boldsymbol{u}}_k\}_{k=1}^{K}$}
\begin{equation}
\bar{\boldsymbol{u}}_k^{(i+1)}=\bar{\boldsymbol{u}}_k^{(i)}+\bar{\boldsymbol{e}}_k^{(i+1)}-\boldsymbol{b}^{(i+1)}, k=1,\dots,K.
\end{equation}

The above formulas iterate and finally yield a suboptimal solution of (P17), denoted by $\tilde{\boldsymbol{b}}$.  Then, since $\tilde{\boldsymbol{b}}=[(t\boldsymbol{v})^T, t]^T$, the corresponding $\tilde{\boldsymbol{v}}=\left[\tilde{\boldsymbol{b}}\right]_{(1:N)} \bigg{ /}\left[\tilde{\boldsymbol{b}}\right]_{N+1}$, $\tilde{\boldsymbol{\Psi}}=\textrm{diag}\left(\left(\tilde{\boldsymbol{v}}\right)^{\star}\right)$.

\subsection{Overall Algorithm}

We summarize the proposed SPMC-SCA-ADMM algorithm in Algorithm 3. The convergence of this algorithm can be readily shown since the objective value of (P2) is monotonically non-decreasing in the iterative process. 

The complexity of Algorithm 3 is mainly due to updating $\boldsymbol{x}$ and $\boldsymbol{b}$ by using (54) and (71). The complexity of updating $\boldsymbol{x}$ is $\mathcal{O}(M^2 + KM)$ and the complexity of updating b is $\mathcal{O}(N^2 + KN)$. Therefore, the overall complexity of Algorithm 3 is $\mathcal{O}(I_{sca}(I_{x}(M^2 + KM) + I_{\Psi}(N^2 + KN)))$, where $I_{x}$ and $I_{\Psi}$ denote the numbers of ADMM iterations required to solve (P11) and (P17) respectively, and $I_{sca}$ denotes the maximum outer iteration number.

\begin{algorithm}[h]
	\caption{SPMC-SCA-ADMM Algorithm}
	\label{alg::conjugateGradient}
	\begin{algorithmic}[1]
		\Require
		$\boldsymbol{H}_d, \boldsymbol{H}_r, \boldsymbol{G}, P$.
		\Ensure
		solution $\{\boldsymbol{x}^{op}$, $\boldsymbol{\Psi}^{op}\}$.
		\State Initialize $\hat{\boldsymbol{x}}^{(0)}$ and $\hat{\boldsymbol{v}}^{(0)}$ to feasible values, initialize $\hat{\boldsymbol{\Psi}}^{(0)}=\textrm{diag}\left(\hat{\boldsymbol{v}}^{(0)\star}\right)$, iteration number $i=0$, $\rho > 0$, $\bar{\rho} > 0$ and threshold $\epsilon > 0$.
		\Repeat
		\State Initialize $\boldsymbol{H}=\boldsymbol{H}_{r} \hat{\boldsymbol{\Psi}}^{(i)} \boldsymbol{G} + \boldsymbol{H}_{d}$ and iteration number $t_1=0$, initialize $\boldsymbol{x}^{(0)},\{\boldsymbol{e}_k^{(0)}\}_{k=1}^{K}$ and $\{\boldsymbol{u}_k^{(0)}\}_{k=1}^{K}$ to feasible values.\footnotemark[1]
		\Repeat
		\State Update $\boldsymbol{x}^{(t_1+1)}$ with (54).
		\State Update $\{\boldsymbol{e}_k^{(t_1+1)}\}_{k=1}^{K}$ with (58).
		\State Update $\{\boldsymbol{u}_k^{(t_1+1)}\}_{k=1}^{K}$ with (60).
		\State Update $t_1=t_1+1$.
		\Until{Convergence or the maximum number of iterations is reached.}
		\State $\hat{\boldsymbol{x}}^{(i+1)}=\boldsymbol{x}^{(t_1)}$.
		\State Initialize
		$\boldsymbol{c}_k=\textrm{diag}(\boldsymbol{h}_{r,k}^{\star})\boldsymbol{G}\hat{\boldsymbol{x}}^{(i+1)}$, $a_k=\boldsymbol{h}_{d,k}^H\hat{\boldsymbol{x}}^{(i+1)}$, $\boldsymbol{l}_k=\left[\boldsymbol{c}_k^T, a_k\right]^T$, $\boldsymbol{L}=\sum_{k=1}^K \boldsymbol{l}_k \boldsymbol{l}_k^H$ and iteration number $t_2=0$, initialize $\boldsymbol{b}^{(0)},\{\bar{\boldsymbol{e}}_k^{(0)}\}_{k=1}^{K}$ and $\{\bar{\boldsymbol{u}}_k^{(0)}\}_{k=1}^{K}$ to feasible values.\footnotemark[1]
		\Repeat
		\State Update $\boldsymbol{b}^{(t_2+1)}$ with (71).
		\State Update $\{\bar{\boldsymbol{e}}_k^{(t_2+1)}\}_{k=1}^{K}$ with (73).
		\State Update $\{\bar{\boldsymbol{u}}_k^{(t_2+1)}\}_{k=1}^{K}$ with (74).
		\State Update $t_2=t_2+1$.
		\Until{Convergence or the maximum number of iterations is reached.}
		\State 
		$\hat{\boldsymbol{v}}^{(i+1)}=\left[\boldsymbol{b}^{(t_2)}\right]_{(1:N)}\bigg{/}\left[\boldsymbol{b}^{(t_2)}\right]_{N+1}$ and $\hat{\boldsymbol{\Psi}}^{(i+1)}=\textrm{diag}\left(\hat{\boldsymbol{v}}^{(i+1)\star}\right)$.
		\State Update $i=i+1$.
		\Until{the fractional increase of the objective value is below the threshold $\epsilon$ or the maximum number of iterations is reached.}
	\end{algorithmic}
\end{algorithm}

\footnotetext[1]{\label{initialization_method}More details of the initialization method can be found in \cite{huang2016consensus}.}

\section{Performance Analysis for Sum Power Maximization}

\subsection{Multiuser Power Scaling Law with Infinitely Large RIS}

In this subsection, we characterize the multiuser power scaling law of the total received power in terms of the number of RIS elements, $N$, as $N \to \infty$. Note that the single-user power scaling law has been studied in \cite{wu2019intelligent}. In the single-user case, the expression of the optimal solution is known, so the power scaling law is not difficult to obtain. However, in the multiuser case, since the optimal solution is not known, it is difficult to directly analyze the power scaling law. Below, we show that the single-user power scaling law is still valid in the multiuser case.

 To highlight the role of RIS, we consider a wireless power transfer system without the BS-user direct channel. For simplicity, assume that the amplitudes of reflection coefficients of all the RIS elements are the same, and that $\boldsymbol{G}$ is a MIMO channel with only the line-of-sight path between the antennas. Without loss of generality, set
\begin{equation} 
\begin{split}
\boldsymbol{G}=\varrho_g e^{-\frac{j2\pi d}{\lambda_c}}\boldsymbol{e}_r \boldsymbol{e}_t^H,
\end{split}
\end{equation}
where
\begin{equation} \nonumber
\boldsymbol{e}_r=
\begin{bmatrix}
1\\
e^{-j2\pi\Delta_r}\\
e^{-j2\pi(2\Delta_r)}\\
\vdots\\
e^{-j2\pi((N-1)\Delta_r)}
\end{bmatrix}, 
\quad
\boldsymbol{e}_t=
\begin{bmatrix}
1\\
e^{-j2\pi\Delta_t}\\
e^{-j2\pi(2\Delta_t)}\\
\vdots\\
e^{-j2\pi((M-1)\Delta_t)}.
\end{bmatrix},
\end{equation}
$\lambda_c$ is the carrier wavelength, $d$ denotes the distance between BS and RIS, $\Delta_r$ and $\Delta_t$ are constants calculated using Angle of Arrival (AoA), Angle of Departure (AoD), carrier wavelength and antenna separation. In this system setting, the total received power is given by $Q=\|\boldsymbol{H}_{r} \boldsymbol{\Psi} \boldsymbol{G}\boldsymbol{x}\|^2$ and the maximum total received power $Q^{op}=\mathop{\max}_{\boldsymbol{x}, \boldsymbol{\Psi}}\|\boldsymbol{H}_{r} \boldsymbol{\Psi} \boldsymbol{G}\boldsymbol{x}\|^2$. The multiuser power scaling law is described in the following proposition.

\begin{mypro}
	Assume the elements of $\boldsymbol{H}_r$ are independently drawn from a common distribution $\mathcal{CN}\left(0,\varrho_h^2\right)$ and $\boldsymbol{G}=\varrho_g e^{-\frac{j2\pi d}{\lambda_c}}\boldsymbol{e}_r \boldsymbol{e}_t^H$. Then 
\begin{equation}
\begin{aligned} 
\frac{\pi}{4}\varrho_g^2 \varrho_h^2  P M \leq \lim_{N \to +\infty} \frac{\mathbb{E}\left[Q^{op}\right]}{N^2} \leq \varrho_g^2 \varrho_h^2  P M.
\end{aligned}
\end{equation}
\end{mypro} 

\begin{proof}[Proof]
	See Appendix B.
\end{proof}

The multiuser power scaling law of the total received power shows that if the RIS is large enough, we can scale down the transmission power of the BS by a factor of $1/N^2$ without compromising the total received power of all the $K$ users.

\subsection{Quantization Loss due to Discrete Reflection Phases}

Since tunable components with finite phase-shift levels are cheaper, in order to ensure the scalability of the RIS, it is more practical to consider that the phases of the RIS elements can only take discrete values \cite{wu2019beamforming, D12020Hybrid, You2020Channel}. A simple way to obtain the discrete phases is to directly quantify the optimal continuous phases. That is, to obtain the phase-shift matrix with discrete phase values, denoted by $\boldsymbol{\Psi}_q$, we directly quantify the phase of each diagonal element of $\boldsymbol{\Psi}^{op}$ uniformly, where $\boldsymbol{\Psi}^{op}$ refers to the optimal phase-shift matrix that maximizes the total received power.

Specifically, $[\boldsymbol{\Psi}_q]_{n, n}$ only takes $2^b$ discrete values, i.e.,
\begin{equation} 
[\boldsymbol{\Psi}_q]_{n, n} \in \mathcal{F} \triangleq \left\{\textrm{exp}\left(j\frac{(2l+1)\pi}{2^b}\right)\right\}_{l=0}^{2^{b}-1},
\end{equation}
where $b$ denotes the phase resolution in number of bits. We set
\begin{equation} 
[\boldsymbol{\Psi}_q]_{n, n}=e^{j\frac{(2l+1)\pi}{2^b}},
\end{equation}
if
\begin{equation} 
\textrm{arg}\left([\boldsymbol{\Psi}^{op}]_{n, n}\right) \in \left[\frac{2\pi l}{2^b},\frac{2\pi(l+1)}{2^b}\right),\forall n,l.
\end{equation}
The above quantization operation is described by projection $\mathrm{Proj}_{\mathcal{F}}(\cdot)$, that is, $\boldsymbol{\Psi}_q=\mathrm{Proj}_{\mathcal{F}}(\boldsymbol{\Psi}^{op})$. When the input to projection $\mathrm{Proj}_{\mathcal{F}}(\cdot)$ is a vector, it does a separate uniform quantization of each element of the input vector.

There is no doubt that taking discrete phase values will bring some performance loss compared with the case where the phase can be taken continuously. To draw useful insights, we analyze the performance loss in a simple wireless power transfer system, where the BS-user direct channel is ignored and $\boldsymbol{G}$ is a MIMO channel with only the line-of-sight path between the antennas. The analysis results are shown in Proposition 2. Note that similar results have appeared in \cite[Proposition 1]{wu2019beamforming}. Compared with the RIS-aided single-user model considered in \cite{wu2019beamforming}, we consider a more general multiuser model. In the single-user case, since the expression of the optimal solution is known, the analysis is relatively easy. However, it is much more difficult to analyze the performance loss caused by quantification in the multiuser case, since the optimal solution is not known. 

Let $Q_{q}$ and $Q^{op}$ represent the received power of the user after employing $\boldsymbol{v}_q$ and $\boldsymbol{v}^{op}$, respectively. Then, define the average power ratio given by
\begin{equation}
\Delta \triangleq \frac{\mathbb{E}\left[Q_{q}\right]}{\mathbb{E}\left[ Q^{op} \right]}.
\end{equation}
Note that $1/\Delta$ is the performance loss. We are now ready to present our result.
\begin{mypro}
	Assume  that the elements of $\sqrt{N} \boldsymbol{H}_r$ are independently drawn from a common distribution $\mathcal{CN}\left(0,1\right)$, and that $\boldsymbol{G}=\varrho_g e^{-\frac{j2\pi d}{\lambda}}\boldsymbol{e}_r \boldsymbol{e}_t^H$. Then as $\textrm{N} \to \infty$, we have $\Delta \geq \frac{2^{2b}}{\pi^2}sin^2\left(\frac{\pi}{2^b}\right) $.
\end{mypro}

\begin{proof}[Proof]
	See Appendix C.
	
\end{proof}

Proposition 2 gives an upper bound on performance loss caused by quantization. As expected, the upper bound decreases as the phase resolution $b$ increases. It is instructive to check Proposition 2 experimentally. When $b=1$, $\Delta$ converges asymptotically to $4/ \pi^2 \approx -3.92$ dB as $\textrm{N} \to \infty$, which means that compared to employing RIS with continuous phases, the performance loss of employing discrete-phase RIS is close to $4$ dB. However, when we set $b=2$, $\Delta$ converges asymptotically to $8/\pi^2 \approx -0.91$ dB as $N \to \infty$, which means that the performance loss is less than $1$ dB. In practice, in order to reduce costs and improve scalability, the phases of RIS elements usually have a relatively low phase resolution $b$, e.g., $b=2$. Our result gives a performance guarantee: Even if the phase resolution is very low (e.g. $b=2$), the performance loss caused by quantization is still acceptable (less than $1$ dB).

 It is worth mentioning that, while Proposition 2 is derived under a specific system model assumption, numerical results demonstrate that it provides insights for the design of a wide range of RIS aided wireless power transfer systems. We will elaborate on this in Section VI.

\section{Numerical Results}

Our simulation scenario is shown in Fig. 3, consisting of one BS, two RISs and eight users. Define set $\mathcal{U}_1=\{u_1,u_2,u_3,u_4\}$ and set $\mathcal{U}_2=\{u_5,u_6,u_7,u_8\}$. The baseband equivalent channels of both the BS-$\textrm{RIS}_1$ link and the BS-$\textrm{RIS}_2$ link are modeled as Rician fading channels whose Rician factor is $\beta_{g}$. The baseband equivalent channels of the the links from $\textrm{RIS}_i$ to users in $\mathcal{U}_i$ ($i=1, 2$) are also modeled as Rician fading channels, whose Rician factor is $\beta_{hr}$. In the unstated case, we set $\beta_{g}=2$, $\beta_{hr}=2$. As for the baseband equivalent channels of the BS-users link and the links from $\textrm{RIS}_i$ to users in $\mathcal{U}_j$ with $i,j\in \{1,2\}$, $i\not=j$, they are modeled as Rayleigh fading channels since we consider that there are obstacles on these links. We set the passloss factor $n=3$ for all the channels. Assume that the geometric size of the two RISs is much larger than the wavelength, so that RISs can be modeled as specular reflectors \cite{basar2019wireless}. Further, we consider uniform linear arrays (ULAs) at the BS and the two RISs. Specifically, $\textrm{RIS}_1$ is positioned in parallel to the BS antenna array, and we set $\delta_0=\pi /4$, $\delta_1=\pi /4$, $\delta_2=\pi /3$. The other system parameters are given as: $f_c=755\textrm{MHz}$, $\eta=1$, $P = 10$ W, $K=8$, $d_1=8$ m,$d_2=7$ m, $d_3=4$ m and $d_4=5$ m.

 \begin{figure}[h]
	\centerline{\includegraphics[width=1\columnwidth]{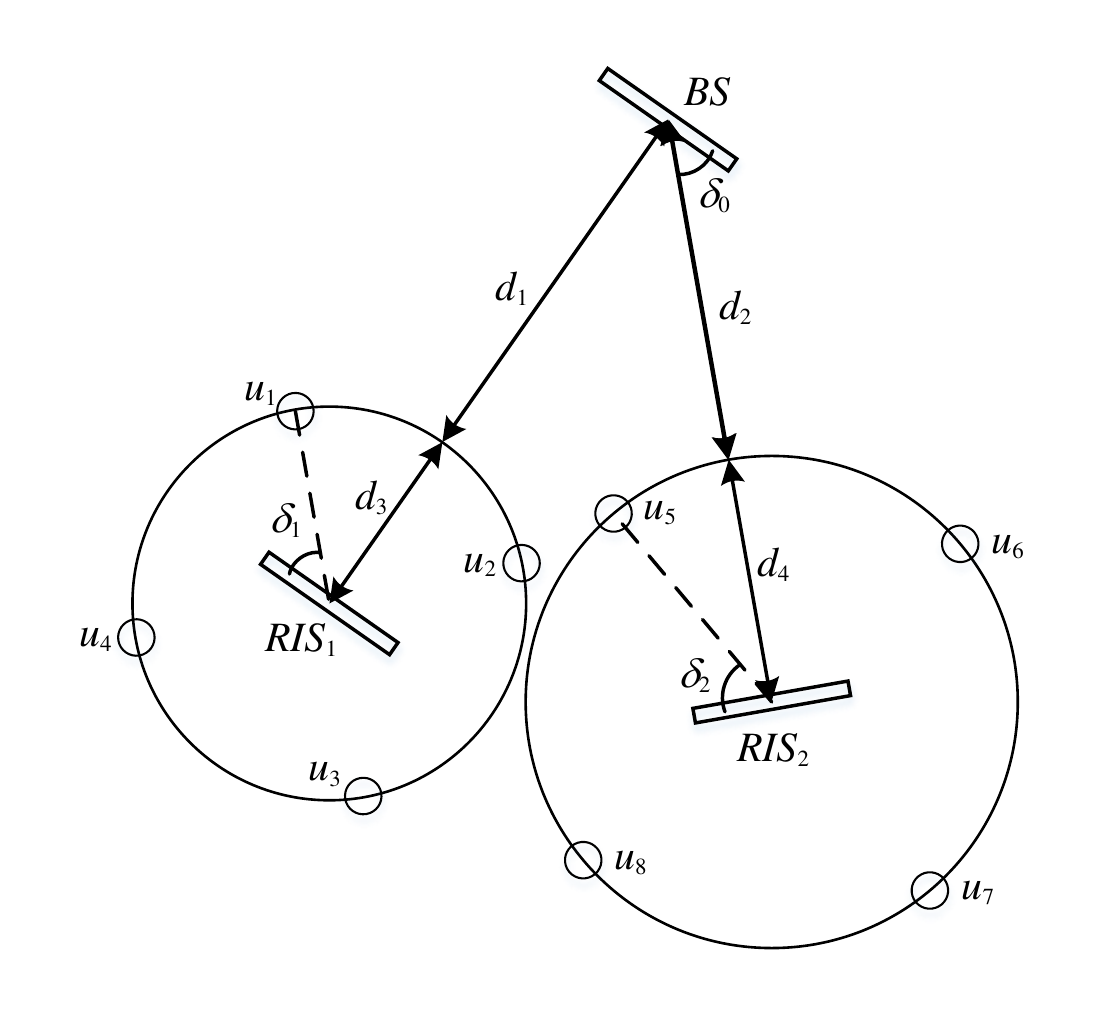}}
	\caption{Simulation setup.}
	\label{1}
\end{figure}

\begin{figure}[htb]
	\centerline{\includegraphics[width=1\columnwidth]{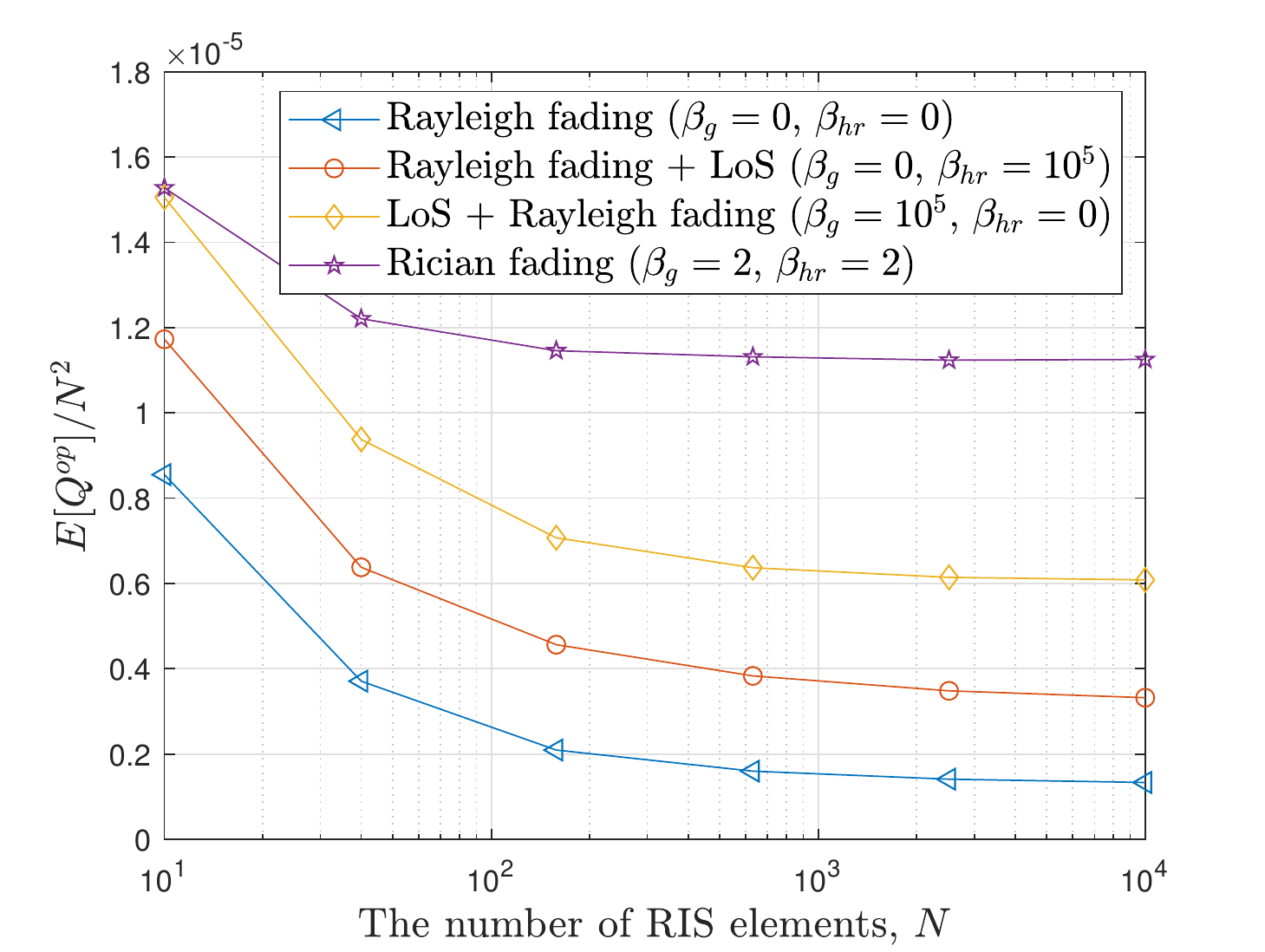}}
	\caption{$\mathbb{E}[Q^{op}]/N$ versus the number of RIS elements, $N$}
	\label{2}
\end{figure}


\begin{figure}[htb]
	\centering
	\includegraphics[width=1\columnwidth]{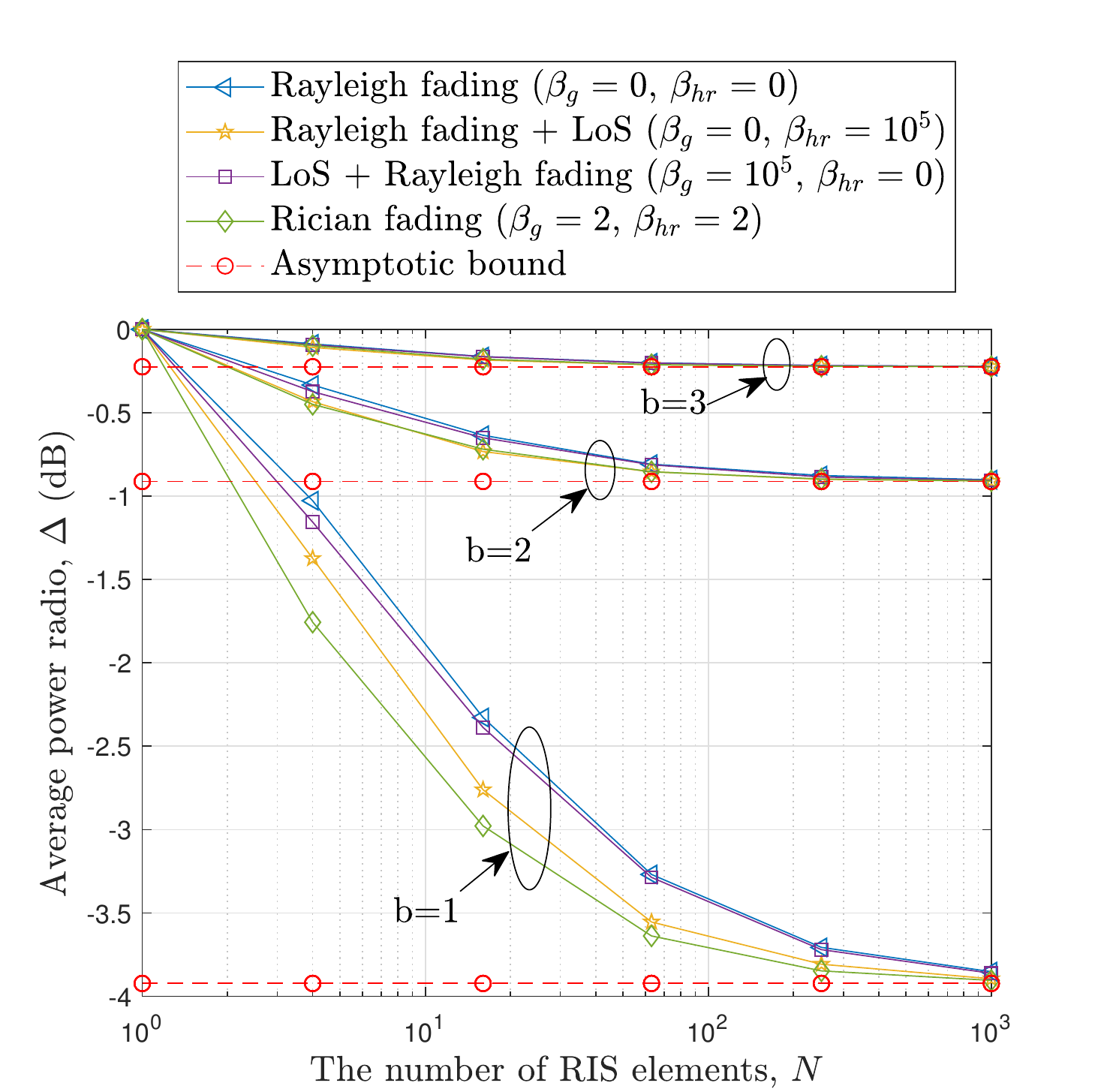}
	\caption{Average power radio versus the number of RIS elements, $N$.}
	\label{3}
\end{figure}



In Fig. 4, we plot $\mathbb{E}[Q^{op}]/N^2$ versus the number of RIS elements with different channel models. Recall that $Q^{op}$ is the optimal objective value of (P1), that is, the maximum total received power. The data required for drawing is obtained by leveraging the SPM-SCA algorithm. To highlight the role of RISs, we ignore the direct channels from the BS to the users. We adopt the scenario in Fig. 3 for simulation where the large-scale fading is considered and the Rician factors $\beta_{g}$ and $\beta_{hr}$ can be adjusted. Note that $\beta_{g} = 0$ means the channels from the BS to the RISs are set as Rayleigh fading channels while $\beta_{g} = 10^5$ and $\beta_{g} = 2$ means they are set as LoS channels and Rician fading channels, respectively; similarly, $\beta_{hr} = 0$, $\beta_{hr} = 10^5$ and $\beta_{hr} = 2$ means the channels from each RIS to the users around it are set as Rayleigh fading channels, LoS channels and Rician fading channels, respectively. In this figure, we observe that as $N \to \infty$, $\mathbb{E}[Q^{op}]/N^2$ tends to a constant under various channel models considered. This means the multiuser power scaling law described in Section V hold not only under the channel model considered in Proposition 1, but also under more practical channel models.

The average power ratio defined in (80) versus the number of RIS elements with different channel models and phase resolutions is shown in Fig. 5. The channel models considered in this figure are exactly the same as that considered in Fig. 5. The asymptotic bounds in this figure are calculated by the expression given in Proposition 2. As shown in the figure, for any phase resolutions ($\boldsymbol{b}=1$, $\boldsymbol{b}=2$ or $\boldsymbol{b}=3$) and channel models (Raleigh fading, Raleigh fading + LoS, LoS + Raleigh fading or Rician fading), as $N \to \infty$, the average power ratio $\Delta$ tends to the corresponding bound. This shows that, although Proposition 2 is derived under a specific channel model assumption, the performance loss upper bound given in Proposition 2 is still valid under a wider range of channel models.

%
%

\begin{figure}[htb]
	\centering
	\includegraphics[width=1\columnwidth]{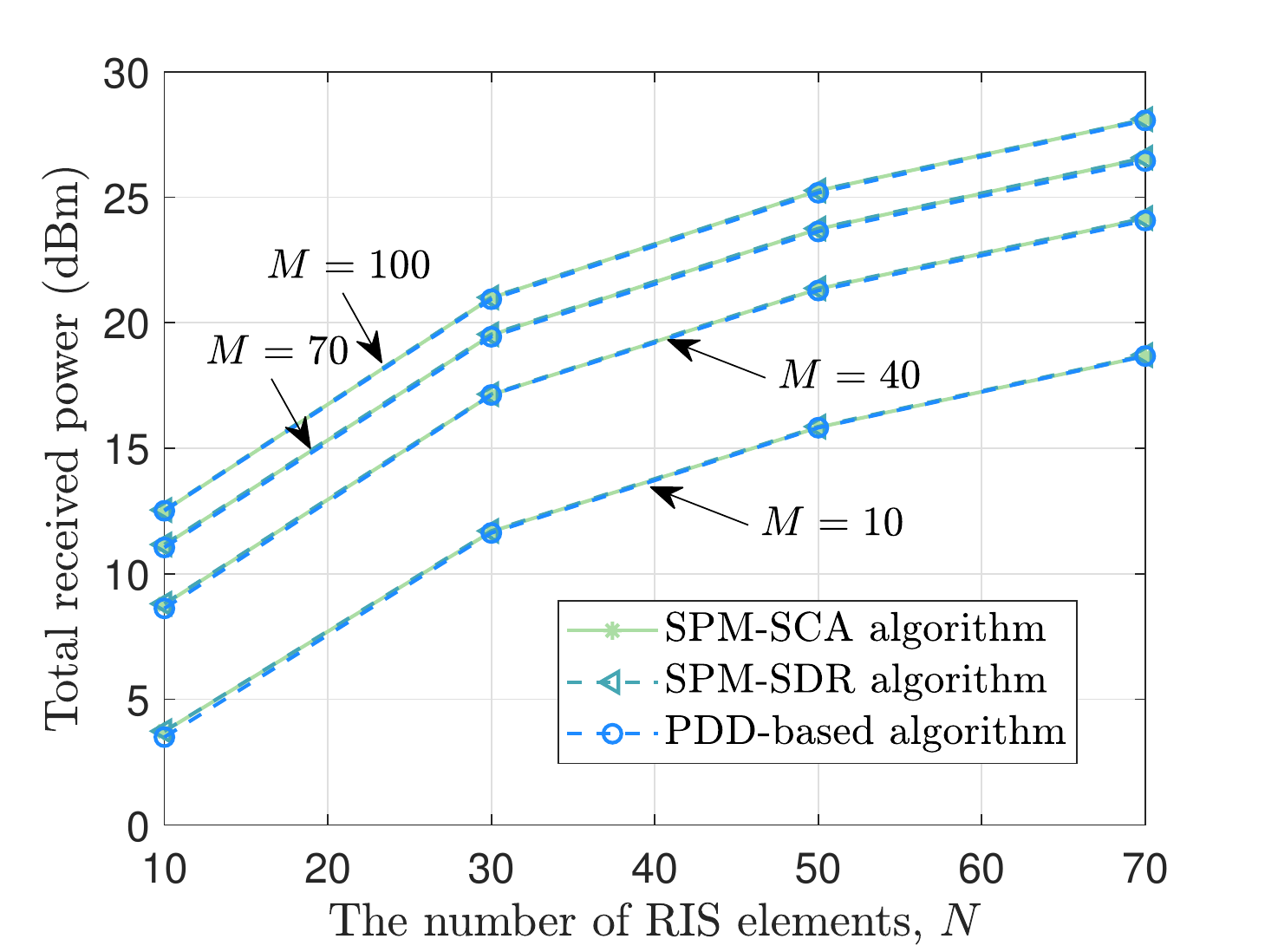}
	\caption{Total received power versus the number of RIS elements, $N$.}
	\label{2}
\end{figure}


In Fig. 6, we plot the total received power versus the number of RIS elements $N$ with different number of BS antennas $M$. Note that the PDD-based algorithm in the figure is proposed in \cite{Zhao2020Intelligent} and can be directly used to solve the two AO subproblems of (P1) (i.e., (P3) and (P5)). We first divide (P1) into two subproblems, and then solve these two subproblems by leveraging the PDD-based algorithm to obtain the curves corresponding to the PDD-based algorithm in the figure. The other two kinds of curves are obtained by leveraging the SPM-SCA and SPM-SDR algorithms proposed in this paper. Naturally, the total received power increases as $N$ increases. Furthermore, the curves corresponding to the three algorithms almost overlap, which indicates that the performance of these three algorithms is almost the same for the problem considered. Note that Table I shows the average execution time of these three algorithms in the cases considered in Fig. 6. From this table, we see that the SPM-SCA algorithm requires the least execution time compared to the benchmark algorithms. Combined with the nearly identical performance shown in Fig. 6, we show that the SPM-SCA algorithm proposed in this paper is an effective low-complexity algorithm.

\begin{figure}[htb]
	\centering
	\includegraphics[width=1\columnwidth]{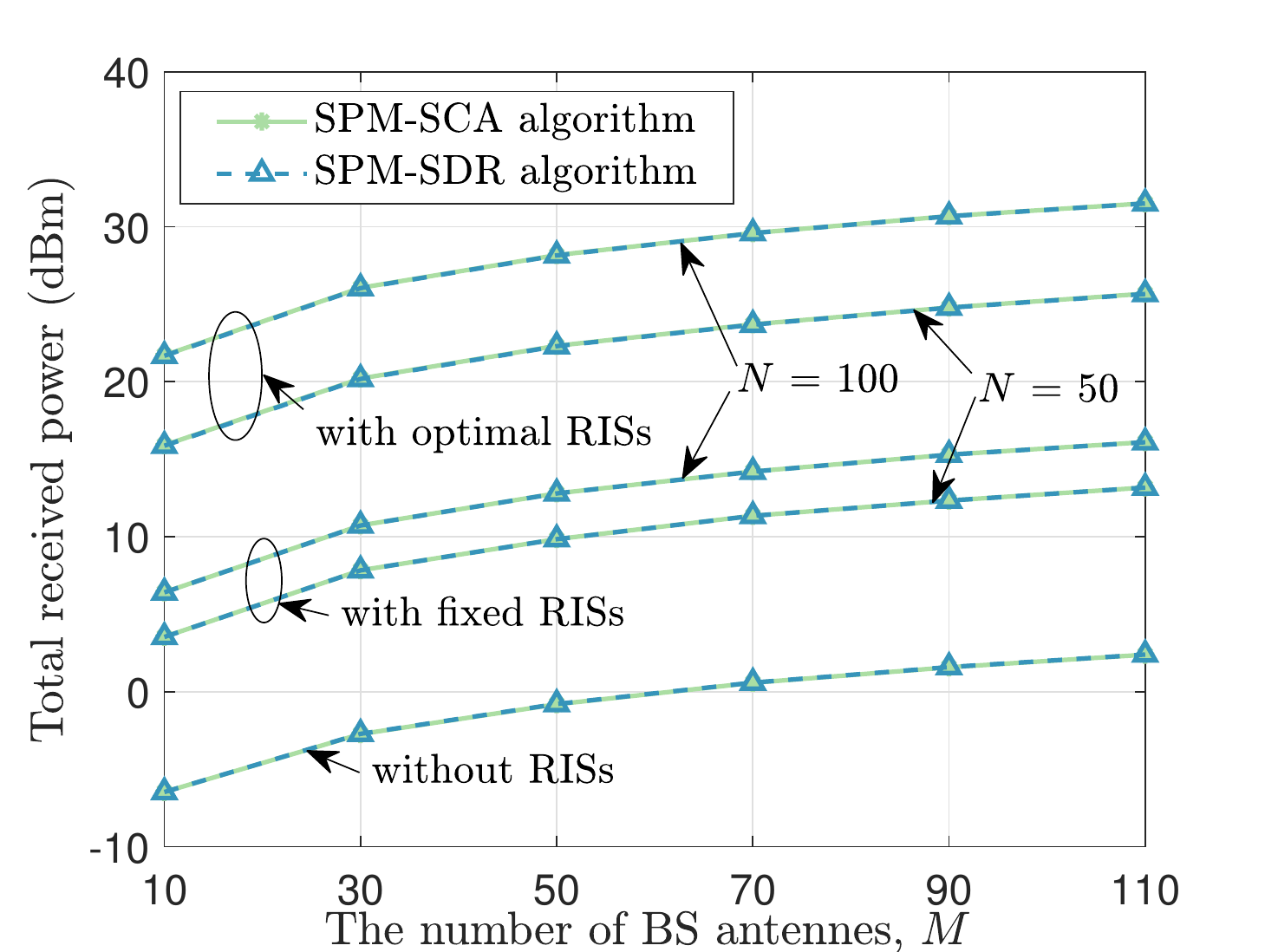}
	\caption{Total received power versus the number of BS antennes, M.}
	\label{2}
\end{figure}

\begin{table}[!h]
	\centering
	\caption{The average execution time of the SPM-SDR, SPM-SCA and PDD-based algorithms.}
	\begin{tabular}{cccccc}
		\hline
		& algorithms & \multicolumn{4}{c}{average execution time ($\times 10^{-1}$ s)} \\ \cline{2-6} 
		&            & N=10           & N=30           & N=50           & N=70           \\ \hline
		\multirow{3}{*}{M=10}  & SPM-SDR    & 128.60         & 137.02         & 145.49         & 160.71         \\ \cline{2-6} 
		& PDD-based  & 3.00           & 4.76           & 7.71           & 12.42          \\ \cline{2-6} 
		& SPM-SCA    & 1.34           & 2.59           & 2.62           & 2.71           \\ \hline
		\multirow{3}{*}{M=40}  & SPM-SDR    & 143.08         & 205.61         & 218.32         & 297.65         \\ \cline{2-6} 
		& PDD-based  & 3.37           & 6.19           & 8.75           & 13.64          \\ \cline{2-6} 
		& SPM-SCA    & 1.81           & 2.10           & 2.01           & 1.93           \\ \hline
		\multirow{3}{*}{M=70}  & SPM-SDR    & 220.14         & 245.38         & 337.98         & 298.75         \\ \cline{2-6} 
		& PDD-based  & 4.74           & 7.30           & 10.25          & 11.59          \\ \cline{2-6} 
		& SPM-SCA    & 1.96           & 1.95           & 1.66           & 1.36           \\ \hline
		\multirow{3}{*}{M=100} & SPM-SDR    & 399.73         & 415.44         & 564.60         & 402.86         \\ \cline{2-6} 
		& PDD-based  & 4.44           & 7.54           & 10.01          & 12.94          \\ \cline{2-6} 
		& SPM-SCA    & 2.99           & 2.09           & 1.41           & 1.27           \\ \hline
	\end{tabular}
\end{table}

Fig. 7 shows the total received power versus the number of BS antennas $M$ under different RIS states. In the case that no RISs are placed in the environment, the corresponding optimization problem is a traditional constant-envelope WPT problem. We solve this problem by leveraging the SCA algorithm. In the case that RISs are fixed, we randomly select the phases of RIS elements in each simulation, and use the SPM-SCA algorithm to optimize the beamformer at the BS. In the case that RISs are placed in the environment and optimized, we use the SPM-SCA algorithm to jointly optimize the beamformer at the BS and the phases of RIS elements. As shown in this figure, the total received power increases as $M$ increases. In the absence of RIS, the total received power is the smallest, while in the presence of RIS, the larger $N$ is, the greater the total received power will be. It is worth noting that, as long as RISs are placed in the environment, even if they are not optimized, the total received power will increase relative to the case without RIS. This makes sense since placing RIS in the environment introduces additional signal propagation paths.



\begin{figure}[htb]
	\centering
	\subfigure[Minimum received power]{
		\begin{minipage}[t]{0.5\columnwidth}
			\centering
			\includegraphics[width=\columnwidth]{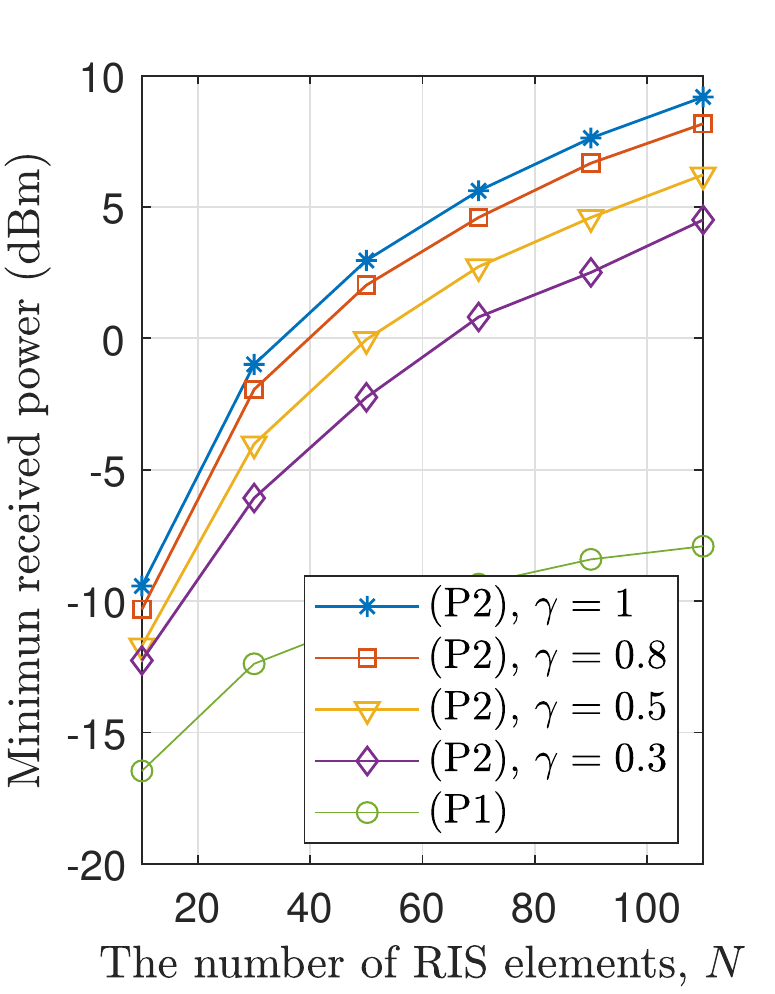}
		\end{minipage}%
	}%
	\subfigure[Total received power]{
		\begin{minipage}[t]{0.5\columnwidth}
			\centering
			\includegraphics[width=\columnwidth]{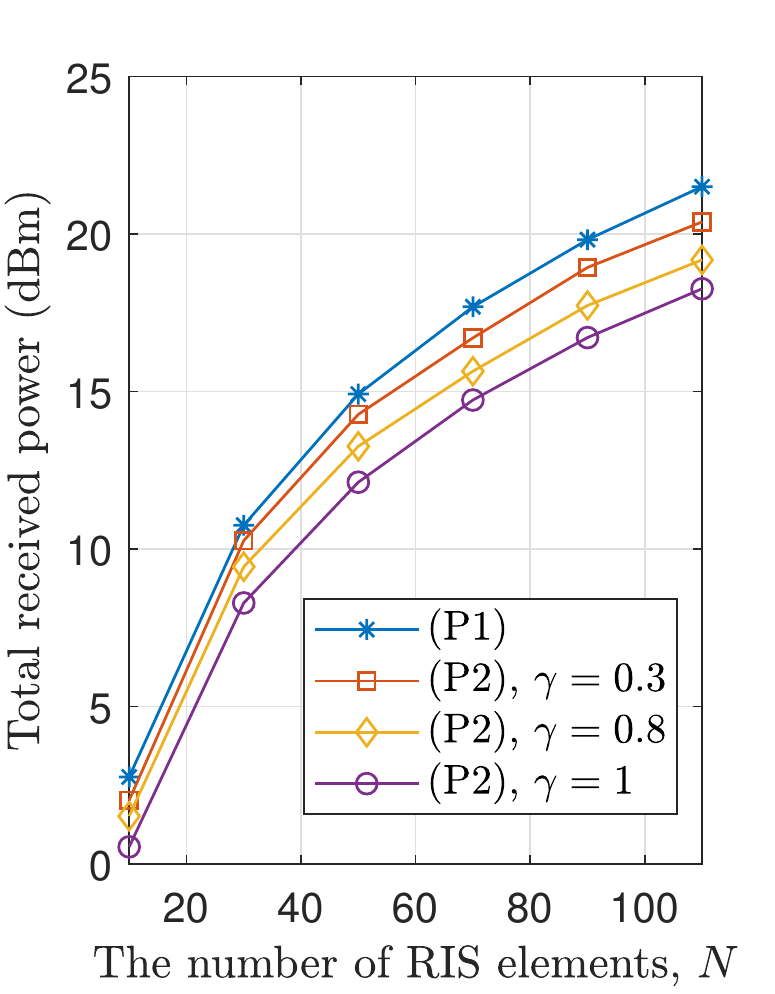}
		\end{minipage}%
	}%
	\centering
	\caption{ Simulation results of (P2).}
\end{figure}

The minimum received power and the total received power versus the number of RIS elements $N$ are shown in Fig. 8(a) and (b), respectively. In the figures, the curves related to (P1) are drawn by SPM-SCA algorithm, while the curves related to (P2) are drawn by SPMC-SCA-ADMM algorithm. Note that we set all the elements in $\left\{p_k/\eta\right\}_{k=1}^K$ to be equal to $\gamma Q_{MM}$ in the simulation, where $\gamma \in [0,1]$ is a scaling factor and $Q_{MM}=\max_{\boldsymbol{\Psi}, \boldsymbol{x}} \min_{1,\dots,K} Q_k$. It is not difficult to see that with this setup, (P2) is solvable. In the simulation, we obtain the value of $Q_{MM}$ under different channel states by trial and error. As can be seen from these two figures, with the increase of $N$, both the total received power and the minimum received power increase. And with the increase of $\gamma$, the total received power decreases while the minimum received power increases. This makes sense since a larger $\gamma$ means tighter constraints and higher requirements for minimum received power. For further explanation, a larger $\gamma$ means a stronger QoS requirement, corresponding to a higher minimum received power in Fig. 8(a), while a stronger QoS requirement leads to a reduction of energy efficiency, corresponding to a smaller total received power in Fig. 8(b). Note that the curves corresponding to (P1) plays the role of bound in both figures. In Fig. 8(a), the curve corresponding to (P1) is a lower bound while in Fig. 8(b), it is a upper bound. As can be seen, with the decrease of $\gamma$, the (P2) related curves in the two figures are getting closer to the (P1) related curves. In fact, when we set $\gamma=0$, (P2) is reduced to (P1).

\section{Conclusion}

In this paper, we proposed a novel scheme to improve the energy efficiency of RIS aided wireless power transfer system. We formulated two different problems, (P1) and (P2), to maximize the total received power of the users by jointly optimizing the beamformer at transmitter and the phase-shifts at the RISs, where the difference is that the solution of (P2) needs to satisfy the user QoS constraints while the solution of (P1) does not. We designed low-complexity algorithms for both problems by applying alternating optimization, SCA and ADMM techniques. Our analysis showed that the average received power increases quadratically with the number of the IRS elements. Our analysis also showed that the power loss caused by RIS phase quantization is not significant. In addition, the effectiveness of the proposed algorithms are verified by extensive simulations.

Looking forward, we list a number of research directions worthy for our future endeavour. First, this paper ignores the signals reflected multiple times by RISs. However, in cases where the cooperation of multiple RISs is required to bypass obstacles, considering the signals reflected by RISs for multiple times will greatly improve the system performance \cite{zheng2020efficient, Mei2020Cooperative}. How to extend the work in this paper to allow multiple RIS reflections still remains an open challenge. Second, this paper assumes perfect CSI. However, accurate estimation of these channels is usually costly. To reduce the channel estimation overhead, in \cite{Zhao2020Intelligent, cai2020two, Zhou2020Robust}, beamforming design for RIS aided communications with imperfect CSI is studied. This inspires us to investigate the beamforming design under imperfect CSI in the RIS aided constant-envelope WPT system.

\appendices
\section{Solution to (P14)}
In this appendix, we derive the solution to (P14). First, we rewrite (P14) as follows:

\begin{alignat}{2}
\textrm{(P14):} \quad \min_{\left\{\boldsymbol{e}_k \right\}_{k=1}^{K}}\quad &
\sum_{k=1}^K \left\|\boldsymbol{e}_k-\boldsymbol{x}^{(i+1)}+\boldsymbol{u}^{(i)}_k\right\|^2 \\
\mbox{s.t.}\quad
&\left|\boldsymbol{h}_k^H \boldsymbol{e}_k\right|^2 \geq \frac{p_k}{\eta}, k=1,\dots,K.
\end{alignat}
Since in the above problem, the optimizations of the elements in set $\left\{\boldsymbol{e}_k\right\}_{k=1}^{K}$ are decoupled and their corresponding optimization problems have the same form, we only need to study one of them as
\begin{alignat}{2}
\textrm{(PA1):} \quad \min_{\boldsymbol{e}_k}\quad & \left\|\boldsymbol{e}_k-\boldsymbol{x}^{(i+1)}+\boldsymbol{u}^{(i)}_k\right\|^2 \\
\mbox{s.t.}\quad
&\left|\boldsymbol{h}_k^H \boldsymbol{e}_k\right|^2 \geq \frac{p_k}{\eta}.
\end{alignat}
For this particular objective function, which is the Euclidean distance from $\boldsymbol{e}_k$ to $\boldsymbol{x}^{(i+1)}-\boldsymbol{u}_k^{(i)}$, the optimum must be on the edge of the feasible region if $\boldsymbol{x}^{(i+1)}-\boldsymbol{u}_k^{(i)}$ is not in the feasible region. Thus, we first check whether $\boldsymbol{x}^{(i+1)}-\boldsymbol{u}_k^{(i)}$ is feasible. If yes, let the optimal $\boldsymbol{e}^{op}_k=\boldsymbol{x}^{(i+1)}-\boldsymbol{u}_k^{(i)}$. Otherwise, solve the following optimization problem instead:
\begin{alignat}{2}
\textrm{(PA2):} \quad \min_{\boldsymbol{e}_k}\quad & \left\|\boldsymbol{e}_k-\boldsymbol{x}^{(i+1)}+\boldsymbol{u}^{(i)}_k\right\|^2 \\
\mbox{s.t.}\quad
&\left|\boldsymbol{h}_k^H \boldsymbol{e}_k\right| = \sqrt{\frac{p_k}{\eta}}.
\end{alignat}

The constraint (86) can be written as a linear constraint with an unknown phase $\nu$:
\begin{equation}
\boldsymbol{h}_k^H\boldsymbol{e}_k=\sqrt{\frac{p_k}{\eta}}e^{j\nu}.
\end{equation}
Suppose we know $\nu$. (PA2) becomes a projection onto an affine subspace \cite{huang2016consensus}, whose solution is given by
\begin{equation}
\boldsymbol{e}_k=\boldsymbol{x}^{(i+1)}-\boldsymbol{u}_k^{(i)}+
\frac{\sqrt{\frac{p_k}{\eta}}e^{j\nu}-\boldsymbol{h}_k^H \left(\boldsymbol{x}^{(i+1)}-\boldsymbol{u}_k^{(i)}\right)}{\left\|\boldsymbol{h}_k\right\|^2}\boldsymbol{h}_k.
\end{equation}
Plugging this intermediate solution to the objective function of (PA2), we see that the minimum is attained if we set $\nu=\textrm{arg}\left(\boldsymbol{h}_k^H \left(\boldsymbol{x}^{(i+1)}-\boldsymbol{u}_k^{(i)}\right)\right)$. Plugging into (88), we obtain
\begin{equation}
\begin{split}
\boldsymbol{e}^{op}_k=&\boldsymbol{x}^{(i+1)}-\boldsymbol{u}_k^{(i)}\\
+&\frac{\sqrt{\frac{p_k}{\eta}}-\left|\boldsymbol{h}_k^H \left(\boldsymbol{x}^{(i+1)}-\boldsymbol{u}_k^{(i)}\right)\right|}{\left\|\boldsymbol{h}_k\right\|^2  \left|\boldsymbol{h}_k^H \left(\boldsymbol{x}^{(i+1)}-\boldsymbol{u}_k^{(i)}\right)\right|}\boldsymbol{h}_k \boldsymbol{h}_k^H  \left(\boldsymbol{x}^{(i+1)}-\boldsymbol{u}_k^{(i)}\right).
\end{split}
\end{equation}
Thus, for $k=1,\dots,K$,
\begin{equation}
\begin{cases}
\boldsymbol{e}_k^{op}=\boldsymbol{x}^{(i+1)}-\boldsymbol{u}_k^{(i)};& \textrm{if} \left|\boldsymbol{h}_k^H\left(\boldsymbol{x}^{(i+1)}-\boldsymbol{u}_k^{(i)}\right)\right|^2 \geq \frac{p_k}{\eta},\\
\textrm{(89)};& \textrm{otherwise}.
\end{cases}
\end{equation}

\section{Proof of Proposition 1}
We prove this proposition by finding a pair of upper and lower bounds of $\mathbb{E}\left[Q^{op}\right]/N^2$, and showing that these two bounds are constants as $N \to +\infty$.

We first simplify the problem form. $Q=\varrho_g^2\|\boldsymbol{H}_{r} \boldsymbol{\Psi} \boldsymbol{e}_r\|^2\left|\boldsymbol{e}_t^H\boldsymbol{x}\right|^2$ can be obtained by replacing $\boldsymbol{G}$ in $\|\boldsymbol{H}_{r} \boldsymbol{\Psi} \boldsymbol{G}\boldsymbol{x}\|^2$ with $\varrho_g e^{-\frac{j2\pi d}{\lambda_c}}\boldsymbol{e}_r \boldsymbol{e}_t^H$. Then we have $Q^{op}=\mathop{\max}_{\boldsymbol{x}, \boldsymbol{\Psi}}\varrho_g^2\|\boldsymbol{H}_{r}$ $ \boldsymbol{\Psi} \boldsymbol{e}_r\|^2\left|\boldsymbol{e}_t^H\boldsymbol{x}\right|^2$, subject to $[\boldsymbol{\Psi}]_{n,n}=1$, $\forall n$ and $[\boldsymbol{x}]_{m}=\sqrt{P/M}$, $\forall m$. It is easy to see that the optimization of the two variables is decoupled. Since the optimal $\boldsymbol{x}$ is clearly given by $\boldsymbol{x}^{op}=\sqrt{P/M}\boldsymbol{e}_t$, the hard part is optimizing $\boldsymbol{\Psi}$. The corresponding optimization problem is $Q^{op}=\mathop{\max}_{\boldsymbol{\Psi}}\varrho_g^2PM\|\boldsymbol{H}_{r}\boldsymbol{\Psi} \boldsymbol{e}_r\|^2$, subject to $[\boldsymbol{\Psi}]_{n,n}=1$. Let $\boldsymbol{z}=\boldsymbol{\Psi}\boldsymbol{e}_r$, we have $\|\boldsymbol{H}_{r} \boldsymbol{\Psi} \boldsymbol{e}_r\|^2=\|\boldsymbol{H}_{r} \boldsymbol{z}\|^2$. Since the module of each element of $\boldsymbol{e}_r$ and each diagonal element of  diagonal matrix $\boldsymbol{\Psi}$ is 1, the module of each element of $\boldsymbol{z}$ is also 1. Then the original priblem can be rewritten as
\begin{alignat}{2}
\textrm{(PB1):} \quad \max_{\boldsymbol{z}} \quad & 
Q^{op}=\varrho_g^2PM\|\boldsymbol{H}_{r} \boldsymbol{z}\|^2 & \\
\mbox{s.t.}\quad
&\left|[\boldsymbol{z}]_n\right|=1, n=1,\dots,N.
\end{alignat}

Then, we give an upper bound and a lower bound of $Q^{op}$. According to \cite[Section \uppercase\expandafter{\romannumeral4}]{zhang2017multi}, we obtain $\varrho_g^2PM N \lambda_{1} $ as an upper bound of $Q^{op}$ by relaxing the constraint of $\boldsymbol{b}$ to $\|\boldsymbol{b}\|^2 \leq N$, and obtain $ \varrho_g^2PM \lambda_{1} \|\boldsymbol{v}_1\|_1^2 $ as a lower bound of $Q^{op}$ by substituting a particular solution, $\textrm{exp}\left(j\textrm{arg}(\boldsymbol{v}_1)\right)$, into the objective function, where $\lambda_{1} $ denotes the maximum eigenvalue of Wishart matrix $\boldsymbol{H}_r^H \boldsymbol{H}_r$, and $\boldsymbol{v}_1$ denotes the normalized eigenvector corresponding to $\lambda_1$. The upper and lower bounds of $Q^{op}$ are shown as 

\begin{equation}
\begin{aligned} 
\varrho_g^2 P M \lambda_{1} \|\boldsymbol{v}_1\|_1^2 \leq Q^{op} \leq \varrho_g^2 P M N \lambda_{1}.
\end{aligned}
\end{equation}
Note that $\boldsymbol{v}_1$ is independent of $\lambda_1$ and $\mathbb{E}\left[\|\boldsymbol{v}_1\|_1^2\right]=1+(N-1)\pi/4$ \cite[APPENDIX]{zhang2017multi}. By simply taking the expectation of the three terms in (93), we have

\begin{equation}
\begin{aligned} 
\left[1+\frac{\pi}{4}(N-1)\right]\varrho_g^2 P M \mathbb{E}\left[\lambda_{1}\right] \leq \mathbb{E}\left[Q^{op}\right] \leq \varrho_g^2 P M N \mathbb{E}\left[\lambda_{1}\right].
\end{aligned}
\end{equation}

Further, according to \cite[Lemma 4.1]{edelman1988eigenvalues}, when $K$ is a constant, $\lim_{N \to +\infty}\mathbb{E}\left[\lambda_{1}\right]/N=\varrho_h^2$. Then we have
\begin{equation}
\begin{aligned} 
\lim_{N \to +\infty} \frac{\left[1+\frac{\pi}{4}(N-1)\right]\varrho_g^2 P M \mathbb{E}\left[\lambda_{1}\right]}{N^2} = \frac{\pi}{4}\varrho_g^2 \varrho_h^2  P M,
\end{aligned}
\end{equation}

\begin{equation}
\begin{aligned} 
\lim_{N \to +\infty} \frac{\varrho_g^2 P M N \mathbb{E}\left[\lambda_{1}\right]}{N^2} = \varrho_g^2 \varrho_h^2  P M.
\end{aligned}
\end{equation}
Combined with (94), we have

\begin{equation}
\begin{aligned} 
\frac{\pi}{4}\varrho_g^2 \varrho_h^2  P M \leq \lim_{N \to +\infty} \frac{\mathbb{E}\left[Q^{op}\right]}{N^2} \leq \varrho_g^2 \varrho_h^2  P M.
\end{aligned}
\end{equation}
This completes the proof.

\section{Proof of Proposition 2}
We first simplify the expression of $\Delta$ as follows. With $\boldsymbol{G}=\varrho_g e^{-\frac{j2\pi d}{\lambda_c}}\boldsymbol{e}_r \boldsymbol{e}_t^H$, we obtain $Q=\varrho_g^2\|\boldsymbol{H}_{r} \boldsymbol{\Psi} \boldsymbol{e}_r\|^2\left|\boldsymbol{e}_t^H\boldsymbol{x}\right|^2$. Then $Q^{op}=\mathop{\max}_{\boldsymbol{x}, \boldsymbol{\Psi}}\varrho_g^2\|\boldsymbol{H}_{r} \boldsymbol{\Psi} \boldsymbol{e}_r\|^2\left|\boldsymbol{e}_t^H\boldsymbol{x}\right|^2$, subject to $[\boldsymbol{\Psi}]_{n,n}=1$, $\forall n$ and $[\boldsymbol{x}]_{m}=\sqrt{P/M}$, $\forall m$. Clearly, the optimal $\boldsymbol{x}$ is given by $\boldsymbol{x}^{op}=\sqrt{P/M}\boldsymbol{e}_t$, and so $Q^{op}=\mathop{\max}_{\boldsymbol{\Psi}}\varrho_g^2PM\|\boldsymbol{H}_{r}\boldsymbol{\Psi} \boldsymbol{e}_r\|^2$, subject to $[\boldsymbol{\Psi}]_{n,n}=1$. Let $\boldsymbol{b}=\textrm{diag}(\boldsymbol{\Psi})$ and $\boldsymbol{H}_e = \boldsymbol{H}_{r} \textrm{diag}\left(\boldsymbol{e}_r\right)$. Then we have $\|\boldsymbol{H}_{r}\boldsymbol{\Psi} \boldsymbol{e}_r\|^2 = \|\boldsymbol{H}_{e} \boldsymbol{b}\|^2$. The original optimization problem can therefore be rewritten as
\begin{alignat}{2}
\textrm{(PC1):} \quad \max_{\boldsymbol{b}} \quad & 
\|\boldsymbol{H}_{e} \boldsymbol{b}\|^2 & \\
\mbox{s.t.}\quad
&\left|[\boldsymbol{b}]_{n}\right|=1, n=1,\dots,N.
\end{alignat}
Thus the average power ratio $\Delta$ can be rewritten as
\begin{equation}
\Delta = \frac{\mathbb{E}\left[\left\|\boldsymbol{H}_{e} \boldsymbol{b}^q\right\|^2\right]}{\mathbb{E}\left[\left\|\boldsymbol{H}_{e} \boldsymbol{b}^{op}\right\|^2\right]},
\end{equation}
where $\boldsymbol{b}^{op}$ is the optimal solution of (PC1) and $\boldsymbol{b}^{q} = \mathrm{Proj}_{\mathcal{F}}(\boldsymbol{b}^{op})$.

We now construct the following optimization problem:
\begin{alignat}{2}
\textrm{(PC2):} \quad \max_{\boldsymbol{b}} \quad & 
\left|\boldsymbol{l}^H \boldsymbol{b}\right|^2 & \\
\mbox{s.t.}\quad
&\left|[\boldsymbol{b}]_{n}\right|=1, n=1,\dots,N,
\end{alignat}
where $\boldsymbol{l} = \boldsymbol{H}_e^H \boldsymbol{H}_e \boldsymbol{b}^{op} / \| \boldsymbol{H}_e^H \boldsymbol{H}_e \boldsymbol{b}^{op}\|$. We have the following result.
\begin{lemma}
	As $N \to \infty$, $\left|\boldsymbol{l}^H \boldsymbol{b}\right|^2 \stackrel{a.s.}{\leq} \left\| \boldsymbol{H}_e \boldsymbol{b}\right\|^2$, where the equality holds for $\boldsymbol{b} = \boldsymbol{b}^{op}$.
\end{lemma}
\begin{proof}
	First, we have 
	\begin{equation}
	\begin{aligned}
	\left|\boldsymbol{l}^H \boldsymbol{b}\right|^2 \!=\! \left|\frac{(\boldsymbol{b}^{op})^H \boldsymbol{H}_e^H}{\| \boldsymbol{H}_e^H \boldsymbol{H}_e \boldsymbol{b}^{op}\|}\boldsymbol{H}_e \boldsymbol{b}\right|^2
	\!\stackrel{(a)}{\leq}\! \frac{ \left\|\boldsymbol{H}_e\boldsymbol{b}^{op}\right\|^2}{\| \boldsymbol{H}_e^H \boldsymbol{H}_e \boldsymbol{b}^{op}\|^2}\| \boldsymbol{H}_e \boldsymbol{b}\|^2,
	\end{aligned}
	\end{equation}
	where step ($a$) is obtained by the Cauchy–Schwarz inequality and the equality holds for $\boldsymbol{b} = \boldsymbol{b}^{op}$.
	
	According to the strong law of large numbers, we have the following element-wise convergence result:
	\begin{equation}
	\begin{aligned}
	\boldsymbol{H}_e \boldsymbol{H}_e^H \xrightarrow[N \to \infty]{a.s.} \boldsymbol{I}_M.
	\end{aligned}
	\end{equation}
	Thus we obtain the following result:
	\begin{equation}
	\begin{aligned}
	\frac{ \left\|\boldsymbol{H}_e\boldsymbol{b}^{op}\right\|^2}{\| \boldsymbol{H}_e^H \boldsymbol{H}_e \boldsymbol{b}^{op}\|^2} = \frac{ (\boldsymbol{b}^{op})^H\boldsymbol{H}_e^H\boldsymbol{H}_e\boldsymbol{b}^{op}}{ (\boldsymbol{b}^{op})^H\boldsymbol{H}_e^H(\boldsymbol{H}_e\boldsymbol{H}_e^H)\boldsymbol{H}_e\boldsymbol{b}^{op}} \xrightarrow[N \to \infty]{a.s.} 1.
	\end{aligned}
	\end{equation}
	Combined with (103), as $N \to \infty$, we have
	\begin{equation}
	\begin{aligned}
	\left|\boldsymbol{l}^H \boldsymbol{b}\right|^2 \stackrel{a.s.}{\leq} \left\| \boldsymbol{H}_e \boldsymbol{b}\right\|^2,
	\end{aligned}
	\end{equation}
	where the equality holds for $\boldsymbol{b} = \boldsymbol{b}^{op}$. 
	
	This completes the proof.
\end{proof}

\begin{corollary}
	As $N \to \infty$, $\boldsymbol{b}^{op}$ is the optimal solution of (PC2).
\end{corollary}
\begin{proof}
	Let $\tilde{\boldsymbol{b}}$ be the optimal solution to (PC2). As $N \to \infty$, we obtain the following result:
	\begin{equation}
	\begin{cases}
	\left|\boldsymbol{l}^H \tilde{\boldsymbol{b}}\right|^2 \stackrel{(a)}{\geq} \left|\boldsymbol{l}^H \boldsymbol{b}^{op}\right|^2 \stackrel{(b)}{=} \| \boldsymbol{H}_e \boldsymbol{b}^{op}\|^2 \stackrel{(c)}{\geq} \| \boldsymbol{H}_e \tilde{\boldsymbol{b}}\|^2,\\
	\left|\boldsymbol{l}^H \tilde{\boldsymbol{b}}\right|^2 \stackrel{(d)}{\leq} \| \boldsymbol{H}_e \tilde{\boldsymbol{b}}\|^2.
	\end{cases}
	\end{equation}
	where step ($a$) and ($c$) are obtained from the fact that $\tilde{\boldsymbol{b}}$ and $\boldsymbol{b}^{op}$ are respectively the optimal solution to (PC2) and (PC1), and step ($b$) and ($d$) are obtained by Lemma 1. Then we have $|\boldsymbol{l}^H \tilde{\boldsymbol{b}}|^2 = |\boldsymbol{l}^H \boldsymbol{b}^{op}|^2$ as $N \to \infty$, which completes the proof.

\end{proof}

Let 
\begin{equation}
\tilde{\Delta} =  \frac{\mathbb{E}\left[\left|\boldsymbol{l}^H \boldsymbol{b}^q\right|^2\right]}{\mathbb{E}\left[\left|\boldsymbol{l}^H \boldsymbol{b}^{op}\right|^2\right]}.
\end{equation}
Then we have the following result.

\begin{corollary}
	As $N \to \infty$, $\Delta \geq \tilde{\Delta}$.
\end{corollary}
\begin{proof}
	From Lemma 1, as $N \to \infty$, $\mathbb{E}\left[\left\|\boldsymbol{H}_{e} \boldsymbol{b}^q\right\|^2\right] \geq \mathbb{E}\left[\left|\boldsymbol{l}^H \boldsymbol{b}^q\right|^2\right]$ and $\mathbb{E}\left[\left\|\boldsymbol{H}_{e} \boldsymbol{b}^{op}\right\|^2\right] = \mathbb{E}\left[\left|\boldsymbol{l}^H \boldsymbol{b}^{op}\right|^2\right]$.
	Thus we have
	\begin{equation}
	\Delta = \frac{\mathbb{E}\left[\left\|\boldsymbol{H}_{e} \boldsymbol{b}^q\right\|^2\right]}{\mathbb{E}\left[\left\|\boldsymbol{H}_{e} \boldsymbol{b}^{op}\right\|^2\right]} \geq  \frac{\mathbb{E}\left[\left|\boldsymbol{l}^H \boldsymbol{b}^q\right|^2\right]}{\mathbb{E}\left[\left|\boldsymbol{l}^H \boldsymbol{b}^{op}\right|^2\right]} = \tilde{\Delta}.
	\end{equation}
\end{proof}

To prove Proposition 2, we need the following lemma:
\begin{lemma}
	As $N \to \infty$, $\tilde{\Delta} \to \frac{2^{2b}}{\pi^2}sin^2\left(\frac{\pi}{2^b}\right)$. 
\end{lemma}
\begin{proof}
	 We start with the distribution of random vector $\boldsymbol{l}$.	Recall that $\boldsymbol{H}_e = \boldsymbol{H}_{r} \textrm{diag}\left(\boldsymbol{e}_r\right)$. Since $\boldsymbol{e}_r$ is a constant vector and the module of each element of $\boldsymbol{e}_r$ is 1, $\boldsymbol{H}_e$ has the same statistical properties as $\boldsymbol{H_r}$, that is, the elements of $\sqrt{N} \boldsymbol{H}_e$ are independently drawn from a common distribution $\mathcal{CN}\left(0,1\right)$. Thus we have $\sqrt{N}\boldsymbol{h}_{m} \sim \mathcal{CN}\left(\boldsymbol{0},\boldsymbol{I}_M\right)$, where $\boldsymbol{h}_m$ is the $m$-th column of $\boldsymbol{H}_e^H$.
	 Note that $\boldsymbol{l}$ can be written as $\boldsymbol{l} =\sum_{m=1}^M \alpha_m\boldsymbol{h}_m$, where $\alpha_m = \boldsymbol{h}_m^H \boldsymbol{b}^{op}/ \| \sum_{m=1}^M \boldsymbol{h}_m^H \boldsymbol{b}^{op} \boldsymbol{h}_m\|$. Thus we have $\sqrt{N}\boldsymbol{l}\sim \mathcal{CN}\left(\boldsymbol{0},\sum_{m=1}^M |\alpha_m|^2\boldsymbol{I}_M\right)$. Moreover, from (105), it is not difficult to show that, as $N \to \infty$, $\sum_{m=1}^M |\alpha_m|^2= \left\|\boldsymbol{H}_e\boldsymbol{b}^{op}\right\|^2/\| \boldsymbol{H}_e^H \boldsymbol{H}_e \boldsymbol{b}^{op}\|^2\xrightarrow[]{a.s.} 1$.

	  Next, we expand $\mathbb{E}\left[\left|\boldsymbol{l}^H \boldsymbol{b}^{op}\right|^2\right]$ and $\mathbb{E}\left[\left|\boldsymbol{l}^H \boldsymbol{b}^{q}\right|^2\right]$. Let $\varphi_n^{op} = \textrm{arg}([\boldsymbol{b}^{op}]_n)$, $\varphi_n^{q} = \textrm{arg}([\boldsymbol{b}^q]_n)$, $\chi_{n} = \textrm{arg}([\boldsymbol{l}]_n)$ and $l_n = \left|[\boldsymbol{l}]_n\right|$. Note that, from Corollary 1, $\boldsymbol{b}^{op}$ is the optimal solution of (PC2) as $N \to \infty$, and by inspection, the optimal solution of (PC2) is $\textrm{exp}(j\textrm{arg}(\boldsymbol{l}))$. Thus, as $N \to \infty$, we have $\varphi_n^{op}=\textrm{arg}([\boldsymbol{b}^{op}]_n)= \textrm{arg}([\boldsymbol{l}]_n)= \chi_{n}$, $\forall n$.
	  Then
	  \begin{equation}
	  \begin{aligned} 
	  &\ \ \ \ \mathbb{E}\left[\left|\boldsymbol{l}^H \boldsymbol{b}^{op}\right|^2\right]\\
	  &=\mathbb{E}\left[\sum_{n=1}^N l_n^2   +  \sum_{n=1}^N l_n  \sum_{i=1, i \neq n}^{N} l_i e^{j\left(\varphi_n^{op} - \chi_{n} - (\varphi_i^{op} - \chi_{i})\right)}\right]\\
	  &=\sum_{n=1}^N \mathbb{E}\left[l_{n}^2\right] \\
	  &\ \ \ + 2 \sum_{n=1}^{N-1} \sum_{i=n+1}^{N} \mathbb{E}\left[l_{n}l_{i}\right] \mathbb{E}\left[ cos\left(\varphi_n^{op}- \chi_{n} - (\varphi_i^{op}-\chi_{i})\right)\right]\\
      &\xrightarrow[N \to \infty]{(a)}\sum_{n=1}^N \mathbb{E}\left[l_{n}^2\right] + 2 \sum_{n=1}^{N-1} \sum_{i=n+1}^{N} \mathbb{E}\left[l_{n}l_{i}\right],
	  \end{aligned}
	  \end{equation}
	  where step ($a$) is obtained by the fact that, as $N \to \infty$, $\varphi_n^{op} = \chi_{n}$, $\forall n$.
	  
	  Similarly, we have
	  \begin{equation}
	  \begin{aligned} 
	  \mathbb{E}\left[\left|\boldsymbol{l}^H \boldsymbol{b}^{q}\right|^2\right]
	  \xrightarrow{N \to \infty}&\sum_{n=1}^N \mathbb{E}\left[l_{n}^2\right]\\
	  & + 2 \sum_{n=1}^{N-1} \sum_{i=n+1}^{N} \mathbb{E}\left[l_{n}l_{i}\right] \mathbb{E}\left[cos\left(e_n - e_i\right)\right].
	  \end{aligned}
	  \end{equation}
	  where quantization error $e_n = \varphi_n^q - \varphi_n^{op}$, $\forall n$. 
	  
	  Since $\varphi_n^{op} = \chi_{n}\sim \mathcal{U}(0, 2\pi)$ and $\varphi_n^{q}$ is the result of uniform quantization of $\varphi_n^{op}$, we have $e_n \sim \mathcal{U}\left(-\pi/2^b,\pi/2^b\right)$, $\forall n$. Since $\varphi_n^{op}$ is independent of $\varphi_i^{op}$, we have $e_n$ is independent of $e_i$, $\forall n \neq i$. As a result,
	  \begin{equation}
	  \begin{aligned} 
	  &\mathbb{E}\left[ cos\left(e_n-e_i\right)\right]\\
	  =&\mathbb{E}_{e_i}\left[\mathbb{E}_{e_n}\left[ cos\left(e_n-e_i | e_i\right)\right]\right]\\
	  =&\int_{-\frac{\pi}{2^b}}^{\frac{\pi}{2^b}}\frac{2^b}{2\pi}\int_{-\frac{\pi}{2^b}}^{\frac{\pi}{2^b}}\frac{2^b}{2\pi}cos\left(r_n-r_i\right)d r_n d r_i\\
	  =&\frac{2^{2b}}{\pi^2}sin^2\left(\frac{\pi}{2^b}\right).
	  \end{aligned}
	  \end{equation}
	  Thus, combining (110), (111) and (112), we have
	  \begin{equation}
	  \begin{aligned}
	  \tilde{\Delta}&\!=\!\frac{\sum_{n=1}^N \mathbb{E}\left[l_{n}^2\right] \!+\! 2 \sum_{n=1}^{N-1}\! \sum_{i=n+1}^{N} \mathbb{E}\left[l_{n}l_{i}\right] \mathbb{E}\left[cos\left(e_n \!-\! e_i\right)\right]}
	  {\sum_{n=1}^N \mathbb{E}\left[l_{n}^2\right] + 2 \sum_{n=1}^{N-1} \sum_{i=n+1}^{N} \mathbb{E}\left[l_{n}l_{i}\right]}\\
	  =&\frac{\sum_{n=1}^N \mathbb{E}\left[l_{n}^2\right] \!+\! 2 \sum_{n=1}^{N-1}\! \sum_{i=n+1}^{N} \mathbb{E}\left[l_{n}l_{i}\right] \frac{2^{2b}}{\pi^2}sin^2\left(\frac{\pi}{2^b}\right)}
	  {\sum_{n=1}^N \mathbb{E}\left[l_{n}^2\right] + 2 \sum_{n=1}^{N-1} \sum_{i=n+1}^{N} \mathbb{E}\left[l_{n}l_{i}\right]}.
	  \end{aligned}
	  \end{equation}
	  As $N \to \infty$, $\sum_{n=1}^N \mathbb{E}\left[l_{n}^2\right] = \mathcal{O}(N)$ and $\sum_{n=1}^{N-1}\! \sum_{i=n+1}^{N} \mathbb{E}\left[l_{n}l_{i}\right] = \mathcal{O}(N^2)$. Then
	  \begin{equation}
	  \tilde{\Delta} \to \frac{2^{2b}}{\pi^2}sin^2\left(\frac{\pi}{2^b}\right).
	  \end{equation}

\end{proof}

Combining Lemma 2 and Corollary 2, we finally obtain that, as $N \to \infty$, $\Delta \geq \frac{2^{2b}}{\pi^2}sin^2\left(\frac{\pi}{2^b}\right)$,
which completes the proof.

\ifCLASSOPTIONcaptionsoff
  \newpage
\fi



%

\bibliographystyle{IEEEtran}

\bibliography{ref.bib}




%








\end{document}